\documentclass[a4paper]{article}

\usepackage[english]{babel}
\usepackage{hyperref}

\usepackage{amsmath, amsfonts,xfrac, amsthm,amssymb,fancyhdr}

\usepackage{mathtools}

\usepackage{setspace}

\usepackage{geometry}
\usepackage{xypic,colortbl}

\usepackage{thmtools}
\usepackage{thm-restate}

\usepackage{ifthen, sidecap, wrapfig, xspace}

\usepackage{mathrsfs, stmaryrd}

\usepackage{sectsty}

\usepackage{
 calc, enumitem, makeidx,bbold
}

\usepackage[font=small,format=plain,labelfont=sc,textfont=it]{caption}
\usepackage[T1]{fontenc}
\usepackage[utf8]{inputenc}
\usepackage[all]{xy}
\usepackage[geometry]{ifsym}
\usepackage[cmyk]{xcolor}

\usepackage{palatino}
\usepackage{mathpazo}
\usepackage{comment}

\usepackage{todonotes}

\usepackage{textpos}

\newcommand{\Oof}{\mathcal{O}}
\renewcommand{\preceq}{\preccurlyeq}

\newcommand{\ind}[2][]{%
  \mathrel{
    \mathop{
      \vcenter{
        \hbox{\oalign{\noalign{\kern-.3ex}\hfil$\vert$\hfil\cr
              \noalign{\kern-.7ex}
              $\smile$\cr\noalign{\kern-.3ex}}}
      }
    }^{#2}\displaylimits_{#1}
  }
}

\newcommand{\from}{\colon}

\renewcommand{\cal}[1]{\mathcal {#1}}
\newcommand{\CC}{\mathscr C}
\renewcommand{\le}{\leqslant}
\renewcommand{\ge}{\geqslant}
\renewcommand{\phi}{\varphi}

\newcommand{\mathsym}[1]{{}}

{\begin{figure}[#1]
\centering 
\includegraphics[scale=#3]{#2}
}
{\end{figure}}

\newlist{enumeratep}{enumerate}{10}
\setlist[enumeratep]{label=\quad\textit{\arabic*'.},ref=\arabic*',leftmargin=*}

\newenvironment{romanlist'}[0]
{\begin{list}{\makebox[0.5cm][l]{\textit{\roman{enumi}')}}}{\usecounter{enumi}}}
{\end{list}}

\newcommand{\savelabel}[2]{\expandafter\newtoks\csname#1\endcsname
  \global\csname#1\endcsname={#2} \label{#1} #2}
\newcommand{\loadlabel}[1]{\noindent {\bf Lemma~\ref{#1}. } \textit{\the\csname#1\endcsname}
\medskip

}
\renewcommand{\setminus}{-}

\newcommand{\loadlabelthm}[1]{\medskip\noindent {\bf Theorem~\ref{#1}. }
  \noindent  \textit{\the\csname#1\endcsname}
\medskip
}

\newcommand{\loadlabelprop}[1]{\noindent {\bf Proposition~\ref{#1}. }
  \noindent  \textit{\the\csname#1\endcsname}
\medskip

}

\newcommand{\N}{\mathbb{N}}

\renewcommand{\subset}{\subseteq}

\newcommand{\atleast}[1]{{\ge n}}
\newcommand{\less}[1]{{<n}}

	\newcommand{\notacol}[2]{}

\newsavebox{\quoteitbox}

{\hspace*{\fill}{\upshape(\usebox{\quoteitbox})}\end{quote}%
\end{sloppypar}\noindent\ignorespacesafterend}

\newenvironment{quoteit*} 
{\begin{sloppypar}\noindent\slshape\begin{quote}\itshape} 
	{\end{quote}\ignorespaces\end{sloppypar}\noindent\ignorespacesafterend}

\newenvironment{quotetag*}
{~\par
	\begingroup                  
	\begin{equation*}
		 \begin{minipage}[c]{115mm}
			\it\noindent{\par}
}
{
		\end{minipage}
	\end{equation*}
	\endgroup                        
\par
\textnormal
\medskip
}

\newcommand{\Ff}{{\mathcal F}}

\newcommand{\Ll}{{\mathcal L}}

\newcommand{\Pp}{{\mathcal P}}
\newcommand{\Qq}{{\mathcal Q}}
\newcommand{\Rr}{{\mathcal R}}

\newcommand\set[1]{\ensuremath{\{#1\}}}
\newcommand{\setof}[2]{\set{#1\mid#2}}

\DeclareMathOperator{\tp}{tp}
\newcommand{\tup}[1]{{\mathbf{#1}}}

\newtheoremstyle{theoremstyle}
  {3pt}
  {3pt}
  {\itshape}
  {0pt}
  {\bfseries}
  {.}
  {4pt}
  {}

\theoremstyle{theoremstyle}
\newtheorem{theorem}{Theorem}[section]

\newtheorem*{theorem*}{Theorem}

\newtheorem{lemma}[theorem]{Lemma}
\newtheorem{corollary}[theorem]{Corollary}
\newtheorem{proposition}[theorem]{Proposition}
\newtheorem{claim}{Claim}

\newtheoremstyle{remarkstyle}
  {3pt}
  {10pt}
  {}
  {0pt}
  {\itshape}
  {}
  {4pt}
  {\thmname{#1}\thmnumber{ #2}\thmnote{ (#3)}.}

\theoremstyle{remarkstyle}

\newtheoremstyle{definitionstyle}
  {3pt}
  {3pt}
  {}
  {0pt}
  {\bfseries}
  {}
  {4pt}
  {\thmname{#1}\thmnumber{ #2}\thmnote{ (#3)}.}

\theoremstyle{definitionstyle}
\newtheorem{definition}{Definition}

\numberwithin{equation}{section}

\newlength{\wideaslength}

\renewcommand{\subset}{\subseteq}

\newcommand{\dist}{{\mathrm{dist}}}
\newcommand{\Oh}{\mathcal{O}}

\newcommand{\seta}[1]{}

\def\lsim{\mathrel{\rlap{\lower4pt\hbox{\hskip1pt$\sim$}}
    \raise1pt\hbox{$<$}}}                

\definecolor{gray1}{rgb}{0.99,0.99,0.99}
\definecolor{gray2}{rgb}{0.97,0.97,0.97}
\definecolor{gray3}{rgb}{0.95,0.95,0.95}
\definecolor{gray4}{rgb}{0.93,0.93,0.93}
\definecolor{gray5}{rgb}{0.91,0.91,0.91}
\definecolor{gray6}{rgb}{0.89,0.89,0.89}
\definecolor{gray7}{rgb}{0.87,0.87,0.87}
\definecolor{gray8}{rgb}{0.85,0.85,0.85}
\definecolor{gray9}{rgb}{0.83,0.83,0.83}
\definecolor{gray10}{rgb}{0.81,0.81,0.81}
\definecolor{gray20}{rgb}{0.71,0.71,0.71}
\definecolor{gray40}{rgb}{0.51,0.51,0.51}

\renewcommand{\le}{\leqslant}
\renewcommand{\leq}{\leqslant}
\renewcommand{\ge}{\geqslant}
\renewcommand{\geq}{\geqslant}

\def\cqedsymbol{\ifmmode$\lrcorner$\else{\unskip\nobreak\hfil
\penalty50\hskip1em\null\nobreak\hfil$\lrcorner$
\parfillskip=0pt\finalhyphendemerits=0\endgraf}\fi} 
\newcommand{\cqed}{\renewcommand{\qed}{\cqedsymbol}}

\DeclareMathOperator{\ltp}{ltp}

\newcommand{\seqmain}[4]{#1_{#2} #4 \ldots #4 #1_{#3}}
\newcommand{\seqsym}[3]{\seqmain{#1}{1}{#2}{#3}}

\newcommand{\seq}[2]{\seqsym{#1}{#2}{,}}

\let\part\relax\DeclareMathOperator{\part}{part}

\title{Twin-width and types\thanks{This work is a part of project {\sc{BOBR}} that has received funding from the European Research Council (ERC) under the European Union's Horizon 2020 research and innovation programme (grant agreement No 948057).}}
\author{Jakub Gajarsk{\'y}\thanks{University of Warsaw, Poland, \texttt{jakub.gajarsky@mimuw.edu.pl}} \and
Micha{\l} Pilipczuk\thanks{University of Warsaw, Poland, \texttt{michal.pilipczuk@mimuw.edu.pl}}\and Wojciech Przybyszewski\thanks{University of Warsaw, Poland, \texttt{przybyszewski@mimuw.edu.pl}} \and Szymon Toru{\'n}czyk\thanks{University of Warsaw, Poland, \texttt{szymtor@mimuw.edu.pl}}}
\date{}
 
\begin{document}
\maketitle
\thispagestyle{empty}

\begin{abstract}
 We study problems connected to first-order logic in graphs of bounded twin-width.
 Inspired by the approach of Bonnet et al.~[FOCS~2020], we introduce a robust methodology of {\em{local types}} and describe their behavior in contraction sequences --- the decomposition notion underlying twin-width. We showcase the applicability of the methodology by proving the following two algorithmic results. In both statements, we fix a first-order formula $\varphi(x_1,\ldots,x_k)$ and a constant~$d$, and we assume that on input we are given a graph $G$ together with a contraction sequence of width at most $d$. 
 \begin{itemize} 
 \item One can in time $\Oh(n)$ construct a data structure that can answer the following queries in time $\Oh(\log \log n)$: given $w_1,\ldots,w_k$, decide whether $\phi(w_1,\ldots,w_k)$ holds in $G$.
  \item After $\Oh(n)$-time preprocessing, one can enumerate all tuples $w_1,\ldots,w_k$ that satisfy $\phi(x_1,\ldots,x_k)$ in $G$ with $\Oh(1)$ delay.
 \end{itemize}
 In the first case, the query time can be reduced to $\Oh(1/\varepsilon)$ at the expense of increasing the construction time to $\Oh(n^{1+\varepsilon})$, for any fixed $\varepsilon>0$. Finally, we also apply our tools to prove the following statement, which shows optimal bounds on the VC density of set systems that are first-order definable in graphs of bounded twin-width.
 \begin{itemize} 
  \item Let $G$ be a graph of twin-width $d$, $A$ be a subset of vertices of $G$, and $\varphi(x_1,\ldots,x_k,y_1,\ldots,y_l)$ be a first-order formula. Then the number of different subsets of $A^k$ definable by $\phi$ using $l$-tuples of vertices from $G$ as parameters, is bounded by $O(|A|^l)$.
 \end{itemize}
\end{abstract}

\begin{textblock}{20}(-1.9, 2.3)
	\includegraphics[width=40px]{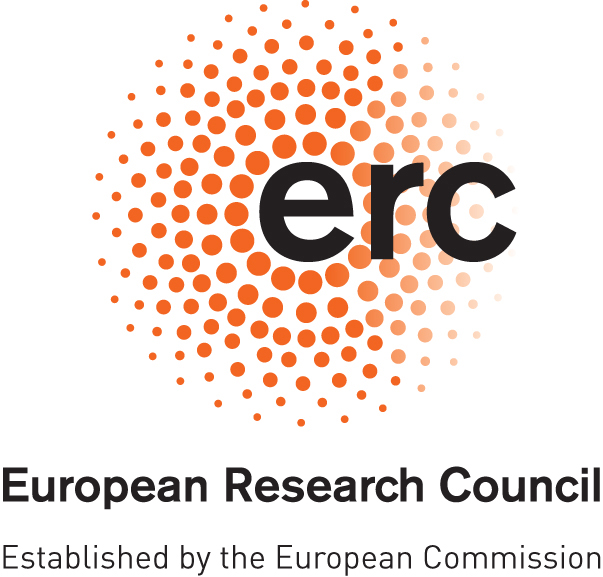}%
\end{textblock}
\begin{textblock}{20}(-1.9,3.2)
	\includegraphics[width=40px]{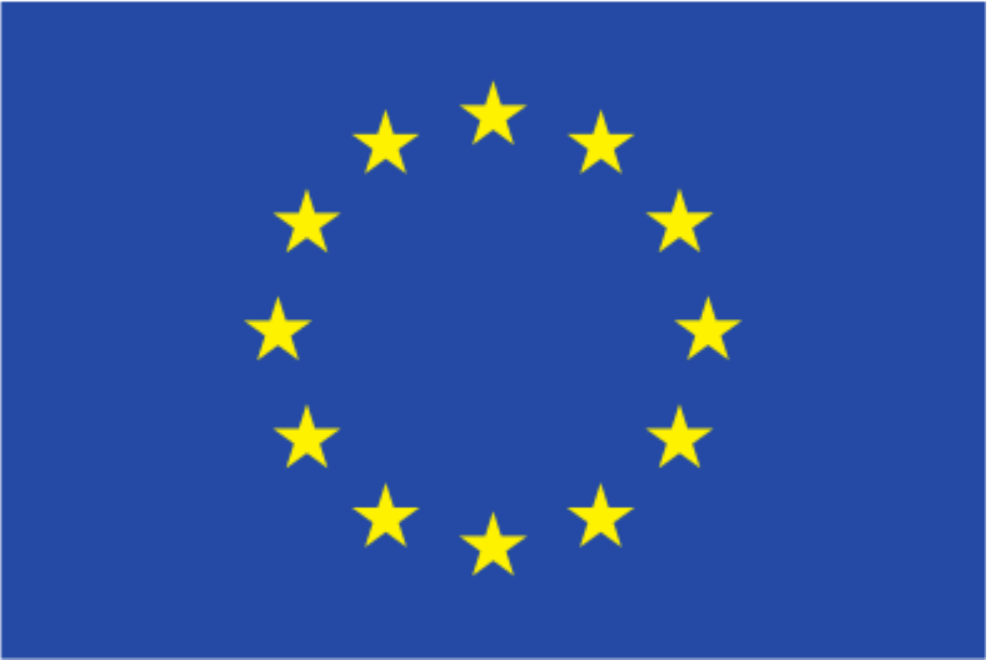}%
\end{textblock}

\newpage

\section{Introduction}\label{sec:intro}

Twin-width is a graph parameter recently introduced by Bonnet et al.~\cite{tww1}. Its definition is based on the concept of a {\em{contraction sequence}}: a sequence of partitions of the vertex set of the graph that starts with the partition into singletons, 
where every subsequent partition is obtained from the previous one by merging two parts ending with the partition with one part.  The main idea lies in measuring the {\em{width}} of a contraction sequence: it is the smallest integer $d$ such that at every step, every part of the current partition is {\em{impure}} --- neither completely adjacent nor completely non-adjacent --- towards at most $d$ other parts of that partition. The {\em{twin-width}} of a graph $G$ is the smallest possible width of a contraction sequence of $G$. Thus, one may think
 that a graph of bounded twin-width can be gradually ``folded'' into a single part so that at every point, every part has a non-trivial interaction with only a bounded number of other parts. 

We remark that while twin-width was originally defined for graphs, the idea can be, and has been, generalized to any classes of binary relational structures, for instance ordered graphs~\cite{tww4} or permutations~\cite{BonnetNOdMST21}. In this work we focus on the graph setting for simplicity. However, all our results lift to arbitrary structures over a fixed relational signature in which all relation symbols have arities at most two.

Since its recent introduction, multiple works have investigated combinatorial, algorithmic, and model-theoretic aspects of twin-width. In this work we are mostly interested in the two last ones. As proved by Bonnet et al.~\cite{tww1}, provided a graph $G$ is given together with a contraction sequence of width bounded by a constant, every property expressible in first-order logic can be verified in linear time on $G$; in other words, the model-checking problem for first-order logic can be solved in linear fixed-parameter tractable time. Further, bounded twin-width is preserved under transductions: if a class of graphs $\CC$ has bounded twin-width, then any class that can be obtained from $\CC$ by a fixed (first-order) transduction also has bounded twin-width~\cite{tww1}. Finally, as proved by Bonnet et al.~\cite{tww4}, classes of {\em{ordered}} graphs that have bounded twin-width exactly coincide with those that are {\em{monadically NIP}}, that is, do not transduce all graphs. All these results witness that twin-width is a model-theoretically important notion and a vital element of the emerging structural theory for graphs based around the notion of a (first-order) transduction. See~\cite{BonnetNOdMST21,GajarskyPT21} for further discussion.

In this work we take a closer look at the model-checking algorithm for graphs of bounded twin-width, proposed in~\cite{tww1}. The basic technical notion used there is that of a {\em{morphism tree}}. While this is not explicit in~\cite{tww1}, it is clear
 that morphism trees are combinatorial objects representing strategies in a form of an Ehrenfeucht-Fra\"isse  game, and basic operations on morphism trees correspond to manipulations on strategies. Mirroring the standard approach taken in finite model theory, one should be able to define a notion of a {\em{type}} suited for the setting of contraction sequences, as well as a corresponding model of an Ehrenfeucht-Fra\"isse game that can be used to argue about properties of types such as compositionality. Providing robust foundations for such a type-based methodology for contraction sequences is the main goal of this work.

We remark that the type/game based perspective of the approach of~\cite{tww1}, which we explained above, was recently briefly outlined in~\cite[Section~5]{tww6}. 

\paragraph*{Our contribution.} We introduce the notion of a {\em{local type}} that is suited for describing first-order properties of tuples of vertices in vertex-partitioned graphs. Intuitively speaking, the rank-$k$ local type of a tuple $\tup w$ in a graph $G$ with vertex partition $\Pp$ is the set of all quantifier rank $k$ formulas satisfied by $\tup w$, where we restrict quantification as follows. Whenever a new vertex, say $z$, is quantified, one has to specify the part $P\in \Pp$ which contains $z$, but at depth $i$ of quantification one can quantify only over parts that are at distance at most $2^{k-i}$ from parts containing already quantified vertices (including vertices of $\tup w$). Here, we mean the distance in the {\em{impurity graph}}: the graph on parts of $\Pp$ where two parts of $\Pp$ are adjacent if and only if they are neither completely adjacent nor completely non-adjacent. This definition mirrors, in logical terms, the morphism trees of Bonnet et al.~\cite{tww1}. In particular, it applies the same idea that the radius of quantification decreases exponentially with the depth. 

We prove a set of fundamental lemmas for manipulation of local types upon consecutive steps in a contraction sequence. These reflect the mechanics of morphism trees of~\cite{tww1}, but by basing the argumentation essentially on Ehrenfeucht-Fra\"iss\'e games, the obtained explanation is arguably simpler and more insightful. Also, contrary to~\cite{tww6,tww1}, the introduced toolbox applies to tuples of vertices, and not only to single parts in the contraction sequence. This is important in our applications, which we discuss next.

We use the toolbox of local types to give the following algorithmic results on first-order expressible problems in graphs of bounded twin-width. The first one concerns the problem of {\em{query answering}}, and the second concerns the problem of {\em{query enumeration}}. In both theorems we assume that the graph is {\em{specified through}} a contraction sequence; this is explained in~Section~\ref{sec:prelims}. For a tuple of parameters $\tup p$, the notation $\Oh_{\tup p}(\cdot)$ hides factors depending on $\tup p$.

\begin{restatable}{theorem}{Answering}\label{thm:intro-main-qa}
 Suppose we are given an $n$-vertex graph $G$ specified through a contraction sequence $\Pp_1,\ldots,\Pp_n$ of width $d$, and a first-order formula $\varphi(\tup x)$, where $\tup x$ is a set of variables. Then one can construct in time $\Oh_{d,\varphi}(n)$ a data structure that can answer the following queries in time $\Oh_{d,\varphi}(\log \log n)$: given $\tup w\in V(G)^{\tup x}$, decide whether $G\models \varphi(\tup w)$. 
\end{restatable}

\begin{restatable}{theorem}{Enumeration}\label{thm:intro-main-qe}
 Suppose we are given an $n$-vertex graph $G$ specified through a contraction sequence $\Pp_1,\ldots,\Pp_n$ of width $d$, and a first-order formula $\varphi(\tup x)$, where $\tup x$ is a set of variables. Then after preprocessing in time $\Oh_{d,\varphi}(n)$, one can enumerate all tuples  $\tup w\in V(G)^{\tup x}$ such that $G\models \varphi(\tup w)$ with $\Oh_{d,\phi}(1)$ delay.
\end{restatable}

Note that in Theorem~\ref{thm:intro-main-qa} there is a factor of the form $\Oh_{d,\varphi}(\log \log n)$ appearing in the query time. This is a consequence of using a data structure for orthogonal range queries of Chan~\cite{Chan11} that supports queries in time $\Oh(\log \log n)$. As explained in~\cite{PilipczukSZ21}, there is also a simple data structure for orthogonal range queries that, for any fixed $\varepsilon>0$, offers query time $\Oh(1/\varepsilon)$ at the expense of increasing the construction time and the space complexity to $\Oh(n^{1+\varepsilon})$. By replacing the usage of the data structure of Chan with this simple data structure, we may obtain the same tradeoff in Theorem~\ref{thm:intro-main-qa}: The query time is reduced to $\Oh(1/\varepsilon)$, while the construction time and the space complexity is increased to $\Oh(n^{1+\varepsilon})$; this holds for any fixed~$\varepsilon>0$.

Theorems~\ref{thm:intro-main-qa} and~\ref{thm:intro-main-qe} mirror classic results on evaluation and enumeration of monadic second-order queries on trees~\cite{Bagan06,colcombet2007factorization,KazanaS13} (which imply analogous results for graphs of bounded treewidth and of bounded cliquewidth), and of first-order queries on classes of bounded expansion~\cite{dvorak2013fo_be,KazanaS19} and nowhere dense classes~\cite{SchweikardtSV18}. We remark that besides Theorems~\ref{thm:intro-main-qa} and~\ref{thm:intro-main-qe}, the toolbox of local types can be also used to reprove in a streamlined way the two fundamental results proved in~\cite{tww1}: the linear-time fixed-parameter tractability of model-checking first-order logic on classes of bounded twin-width, and the stability of twin-width under first-order interpretations. We believe that these results witness the robustness of the developed methodology.

As another application, we prove optimal bounds on {\em{VC density}} of set systems definable in graphs of bounded twin-width. Suppose $\varphi(\tup x,\tup y)$ is a first-order formula with free variables partitioned into $\tup x$ and $\tup y$. For a graph $G$ and a subset of vertices $A$, we define the {\em{Stone space}}
$$S^\varphi(A)\coloneqq \left\{\,\{\,\tup a\in A^{\tup x}~|~G\models \varphi(\tup a,\tup b)\,\}~\colon~\tup b\in V(G)^{\tup y}\,\right\}.$$
In other words, every tuple $\tup b\in V(G)^{\tup y}$ gives rise to the subset $\phi(A,\tup b)\subseteq A^{\tup x}$ consisting of those tuples $\tup a$ that together with $\tup b$ satisfy $\varphi$. Then the Stone space $S^{\varphi}(A)$ consists of all sets $\phi(A,\tup b)$ that can be defined in this way. See for example~\cite{PilipczukST18a} 
 for a discussion of this notion and its applications.

In general graphs, $S^{\varphi}(A)$ can be as large as the whole powerset of $A^{x}$. However, under various structural assumptions, it will be typically much smaller. For instance, suppose that the twin-width of $G$ is bounded by a constant $d$. Then by combining the results of Bonnet et al.~\cite{tww1} with that of Baldwin and Shelah~\cite{baldwinshelah}, one can argue that the VC dimension of $S^{\varphi}(A)$, regarded as a set system over universe $A^{\tup x}$, is bounded by a constant depending only on $d$ and $\varphi$. Consequently, by the Sauer-Shelah Lemma~\cite{sauer1972density, shelah1972combinatorial}, the cardinality of $S^{\varphi}(A)$ is bounded polynomially in $|A|$. 
However, the degree of this polynomial bound, which is known as the {\em{VC density}} (studied for example in~\cite{VC-density}), still depends on $d$ and $\varphi$, and in a quite non-explicit way. 
We prove that in fact, there is a much sharper upper bound: the VC density is bounded by just the number of variables in $\tup y$.

\begin{restatable}{theorem}{VCdensity}\label{thm:intro-main-nei}
 Let $G$ be a graph of twin-width at most $d$, $A$ be a subset of vertices of $G$, and $\varphi(\tup x,\tup y)$ be a first-order formula. Then
 $$|S^{\varphi}(A)|\leq \Oh_{d,\varphi}\left(|A|^{|\tup y|}\right).$$
\end{restatable}

It is easy to see (see e.g.~\cite{PilipczukST18a}) that even in edgeless graphs one cannot hope for a bound better than $|A|^{|\tup y|}$, and therefore the bound of  Theorem~\ref{thm:intro-main-nei} is asymptotically optimum.

Theorem~\ref{thm:intro-main-nei} mirrors analogous results for monadic second-order formulas on classes of bounded treewidth or  cliquewidth~\cite{PaszkeP20}, and for first-order formulas on classes of bounded expansion and nowhere dense classes~\cite{PilipczukST18a}. We remark that the case $|\tup x|=|\tup y|=1$ follows from the fact that classes of bounded twin-width are closed under first-order transductions, combined with known linear upper bounds on the {\em{neighborhood complexity}}\footnote{In our notation, bounds on neighborhood complexity exactly correspond to the case when $\tup x=\{x\}$, $\tup y=\{y\}$, and $\varphi(x,y)$ just checks that $x$ and $y$ are adjacent.} in classes of bounded twin-width~\cite{BonnetKRTW21, przybyszewski2022vcdensity}. Tackling multiple free variables requires a better understanding of types for tuples of vertices, which is exactly where our methodology of local types comes into~play.

\paragraph*{Organization.} After preliminaries in Section~\ref{sec:prelims}, we present the framework of local types in Section~\ref{sec:types}, and immediately derive from it the classic results proved in~\cite{tww1}. Then we prove Theorem~\ref{thm:intro-main-qa} in Section~\ref{sec:queries}, Theorem~\ref{thm:intro-main-qe} in Section~\ref{sec:enumeration}, and
Theorem~\ref{thm:intro-main-nei} in Section~\ref{sec:numtypes}. Appendix~\ref{app:effective} is devoted to the derivation of effective variants of some lemmas from Section~\ref{sec:types}

\paragraph*{Acknowledgements.} The authors thank Rose McCarty and Felix Reidl for many initial discussions on the type approach to first-order logic on graphs of bounded twin-width.

\newcommand{\ImpGraph}{\mathsf{Imp}}
\newcommand{\Vic}{\mathsf{Vicinity}}
\newcommand{\prx}{\mathsf{firstClose}}
\newcommand{\Types}{\mathsf{Types}}
\newcommand{\RTypes}{\mathsf{RTypes}}
\newcommand{\Aff}{\mathsf{Affected}}
\newcommand{\Rel}{\mathsf{Relevant}}

\newcommand{\iwarp}[2]{#1\langle #2\rangle}
\newcommand{\warp}[3]{#1\langle #2\to #3\rangle}

\newcommand{\Tri}{G}
\newcommand{\Imp}{G^{\textsf{imp}}}

\section{Preliminaries}\label{sec:prelims}

\paragraph*{Graphs.} In this paper we work with finite, undirected graphs and we use standard graph notation. By $|G|$ we denote the number of vertices of a graph $G$.

A pair of disjoint vertex subsets $A,B\subset V(G)$ is {\em{complete}} if every vertex of $A$ is adjacent to every vertex of $B$, and {\em{anti-complete}} if there is no edge with one endpoint in $A$ and the other one in $B$. The pair $A,B$ is {\em{pure}} if it is complete or anti-complete, and {\em{impure}} otherwise.

A {\em{trigraph}} is a structure in which there is a vertex set and every pair of distinct vertices is bound by exactly one of the following three symmetric relations: adjacency, non-adjacency, and impurity. Thus, graphs are trigraphs without impurities.
Given a partition of $\Pp$ of the vertex set of a graph $G$, we define the {\em{quotient trigraph}} $G/\Pp$ as the trigraph on vertex set $\Pp$ where distinct $A,B\in \Pp$ are adjacent if the pair $A,B$ is complete in $G$, non-adjacent if the pair is anti-complete, and impure towards each other if $A,B$ is impure. For a trigraph $H$, its {\em{impurity graph}} $\ImpGraph(H)$ is the graph on vertex set $H$ where two vertices $u,v\in V(H)$ are considered adjacent if they are impure towards each other in $H$.

\paragraph*{Contraction sequences.}
Let $G$ be a graph on $n$ vertices.
A \emph{contraction sequence} for $G$
is a sequence $\Pp_1,\ldots,\Pp_n$
of partitions of the vertex set of $G$ such that:
\begin{itemize} 
    \item $\cal P_1$ is the partition into singletons;
     \item $\cal P_n$ is the partition with one part;
    \item for each $t\in [n]$, $t>1$, $\cal P_{t}$ is obtained from $\cal P_{t-1}$ by taking some two parts $A,B\in \cal P_{t-1}$ and {\em{contracting}} them: replacing them with a single part $A\cup B\in \Pp_{t}$.
\end{itemize}
Indices $t\in [n]$ will be called {\em{times}}. The {\em{width}} of the contraction sequence $\Pp_1,\ldots,\Pp_n$ is the maximum degree in graphs $\ImpGraph(G/\Pp_t)$, at all times $t\in [n]$. The {\em{twin-width}} of $G$ is the minimum possible width of a contraction sequence of $G$.

If $G$ is supplied with a total order $\leq$ on $V(G)$, then a subset of vertices $A$ is {\em{convex}} if it forms an interval in $\leq$, that is, if $a,b\in A$ then also $c\in A$ whenever $a\leq c\leq b$. A partition is convex if all its parts are convex, and a contraction sequence is convex if all its partitions are~convex.

\paragraph*{Additional notation for partitions and contraction sequences.}
Fix a graph $G$ with a partition $\Pp$ of its vertices.
We will use the following notation.

Denote $\Tri_\Pp\coloneqq G/\Pp$ and $\Imp_\Pp\coloneqq \ImpGraph(\Tri_\Pp)$. By $\dist_\Pp(\cdot,\cdot)$ we denote the distance function in $\Imp_\Pp$: for $A,B\in \Pp$, $\dist_\Pp(A,B)$ is the minimum length of a path in $\Imp_\Pp$ connecting $A$ and $B$, and $+\infty$ if there is no such path. We extend this notation to subsets, or tuples of elements of $\Pp$: $\dist_\Pp(X,Y)$ denotes the minimum, over all $A$ occurring in $X$ and $B$ occurring in $B$, of $\dist_\Pp(A,B)$.

For a set of parts $\Ff\subseteq \Pp$ and a radius parameter $r\in \N$, the {\em{$r$-vicinity}} of $\Ff$, denoted $\Vic_\Pp^r(\Ff)$, is the trigraph induced in $\Tri_\Pp$ by all parts at distance at most $r$ from any part belonging to $\Ff$, that is
$$\Vic_\Pp^r(\Ff)\coloneqq \Tri_\Pp [\{A\in \Pp~|~\dist_\Pp(A,\Ff)\leq r\}].$$
We may use notation $\Vic_\Pp^r(\cdot)$ for single parts or tuples of parts with the obvious meaning.

\medskip
For brevity, whenever a graph $G$ and its contraction sequence $\Pp_1,\ldots,\Pp_n$ are clear from the context, we fix the following notation. First, in all the notation defined above, concerning partitions, we write $s$ in the subscript instead of $\Pp_s$. So for instance we write $G_s$ to denote $G_{\Pp_s}$, and $\Imp_s$ to denote $\Imp_{\Pp_s}$, and $\dist_s(\cdot,\cdot)$ to denote $\dist_{\Pp_s}(\cdot,\cdot)$, etc.

Fix a finite set of variables $\tup x$. For a pair of times $s,t\in [n]$, $s\leq t$, and a tuple of parts $\tup u\in \Pp_s^{\tup x}$, we define the tuple $\warp{\tup u}{s}{t}\in \Pp_t^{\tup x}$ as follows: for each $y\in \tup x$, $\warp{\tup u}{s}{t}(y)$ is the unique part of $\Pp_t$ that contains $\tup u(y)$. For a tuple $\tup u\in V(G)^{\tup x}$ of vertices and $s\in [n]$, by $\iwarp{\tup u}{s}\in \Pp_s^{\tup x}$ we denote the unique tuple  whose $y$-component, for $y\in\tup x$, is the part of $\Pp_s$ containing $\tup u(y)$.

For $s\in[n-1]$, by $B_{s+1}$ we denote the unique part of $\Pp_{s+1}$ that is the union of two parts in~$\Pp_{s}$.
For a parameter $r\in \N$, we define the {\em{$r$-relevant region}} in $G_s$ as follows:
$$\Rel^r_s\coloneqq G_s[\{C\in \Pp_s~|~C\subseteq B_{s+1},\textrm{ or }C\in \Pp_{s+1}\textrm{ and }\dist_{s+1}(C,B_{s+1})\leq r\}].$$
In other words, $\Rel^r_s$ is the trigraph induced in $G_s$ by the two parts of $\Pp_s$ that get contracted into $B_{s+1}$ and all parts of $\Pp_s$ that stay intact in $\Pp_{s+1}$ and are at distance at most $r$ from $B_{s+1}$ in~$\Imp_{s+1}$.

Note that we have $|\Rel_s^p|\leq \Oh_{d,p}(1)$ for all $s\in [n-1]$. The next lemma shows that the $p$-relevant regions can be computed efficiently. The proof boils down to tracing the trigraph carefully along the contraction sequence.

\begin{lemma}\label{lem:compute-affected}
 Suppose a graph $G$ on vertex set $[n]$ is provided through a convex contraction sequence $\Pp$ of width $d$. Then for a given $p\in \N$, one can in time $\Oh_{d,p}(n)$ compute the trigraphs $\Rel_s^p$ for all~$s\in [n-1]$.
\end{lemma}
\begin{proof}
 For a time $s\in [n]$, $s>1$, let the {\em{$p$-affected region}} $\Aff_s^p$ be the $p$-vicinity of $B_s$ (the part of $\Pp_s$ obtained from the contraction of two parts of $\Pp_{s-1}$). That is, 
 $$\Aff_s^p\coloneqq \Vic^p_s(B_s).$$
 Observe that given $\Aff_{s+1}^p$ and the information associated in the contraction sequence with the contraction at the time $s$, one can compute the $p$-relevant region $\Rel_s^p$ in time $\Oh_{d,p}(1)$. Therefore, from now on we may focus on computing the $p$-affected regions $\Aff_s^p$ for all $s\in [n]$, $s>1$.

 First observe that in time $\Oh_d(n)$ we can scan the contraction sequence while maintaining, at every time $s\in [n]$, the current impurity graph $\Imp_s$ (say, represented through adjacency lists). Indeed, for $s=1$ the graph $\Imp_1$ is edgeless, and updating $\Imp_{s}$ to $\Imp_{s+1}$ using the data provided with the contraction sequence requires time $\Oh_d(1)$. 
 We can therefore execute this scan and for every $s\in [n]$, record the subgraph induced in $\Imp_t$ by parts at distance at most~$p$; call this subgraph $H_s$. Since $|H_s|\leq \Oh_{d,p}(1)$ for each $s\in [n]$, this computation takes total time~$\Oh_{d,p}(n)$. 
 
 To construct trigraphs $\Aff_s$ from graphs $H_s$ it remains to determine, for each $s\in [n]$ and every pair of parts $X,Y\in V(H_s)$ that are non-adjacent in $H_s$, whether $X,Y$ are complete or anti-complete. We do it by a second scan of the contraction sequence as follows.
 
 Let $\Lambda$ be the set of all triples $(X,Y,s)$ as above; note that $\Lambda$ can be computed in time~$\Oh_{d,p}(n)$ by considering every graph $H_s$ separately.
 For every triple $(X,Y,s)\in \Lambda$, let $t(X,Y,s)$ be the earliest time such that the following holds: if $X',Y'\in \Pp_t$ are the unique parts of $\Pp_{t(X,Y,s)}$ containing $X$ and $Y$, respectively, then either $X'=Y'$ or the pair $X',Y'$ is impure. Note that $t(X,Y,s)$ is well-defined, because all vertices are eventually contracted into a single part in $\Pp_n$, and it always holds that $s<t(X,Y,s)$.
 
 The idea is that we scan the contraction sequence again and at every time $t\in [n]$, we determine completeness or anti-completeness of all pairs $A,B$ with $(A,B,s)\in \Lambda$ such that $t=t(A,B,s)$.
 
 Call a pair of distinct parts $C,D\in \Pp_t$ {\em{local}} at the time $t$ if $C,D$ are non-adjacent in $\Imp_t$, but $\dist_t(C,D)\leq p$. Note that there are only $\Oh_{d,p}(|\Pp_t|)$ local pairs at every time $t$. During the scan, with every local pair $C,D\in \Pp_t$ we maintain a list $\Ll_{C,D}$ consisting of all pairs $(X,Y,s)\in \Lambda$ such that $X\subseteq C$, $Y\subseteq D$, and $s\leq t$. With every part $C\in \Pp_t$ we maintain the set of all local pairs $C,D$ involving $C$, each with a pointer to the list $\Ll_{C,D}$ described above. Note that at the time $t=1$ there are no local pairs, so there is no need for initialization. Therefore, we now describe how to update the lists upon moving from the time $t-1$ to the time $t$, and how to use them to resolve the completeness or anti-completeness of all pairs $(A,B,s)$ with $t=t(A,B,s)$.
 
 Suppose at time $t$ one contracts parts $A,A'\in \Pp_{t-1}$ into the part $B\coloneqq B_t=A\cup A'\in \Pp_t$. It can be easily seen that the necessary updates can be done by performing the following operations.
 \begin{itemize}
  \item If the pair $A,A'$ was local at the time  $t-1$ (this can be determined by checking if it has an associated list), then in the data associated with the contraction sequence it is described whether the pair $A,A'$ is complete or anti-complete. For every triple $(X,Y,s)\in \Ll_{A,A'}$, set the relation between $A$ and $A'$ in $\Aff_s$ accordingly. Then destroy the list $\Ll_{A,A'}$.
  \item For every part $C$ that is adjacent to $B$ in $\Imp_t$, check whether the pair $A,C$ was local at the time $t-1$. If so, perform the same operation on pair $A,C$ as was done on pair $A,A'$ in the previous point. Also, do the same for the pair $A',C$.
  \item For every part $D\in \Pp_t$ such that the pair $B,D$ is local at the time $t$, set the list $\Ll_{B,D}$ as follows. Check whether pairs $A,D$ and $A',D$ were local at the time $t-1$. If none of them was, then set $\Ll_{B,D}$ to be an empty list. If one of them was, say $A,D$, then set $\Ll_{B,D}\coloneqq \Ll_{A,D}$. If both of them were, then concatenate $\Ll_{A,D}$ and $\Ll_{A',D}$ and set $\Ll_{B,D}$ to be the obtained list. Finally, in all three cases above, append the triple $(B,D,t)$ to the list $\Ll_{B,D}$.
 \end{itemize}
 The time complexity of operations presented above is $\Oh_{d,p}(1)$ plus the total length of lists destroyed in the first two points. Observe that at every time we construct only $\Oh_{d,p}(1)$ new list elements, hence the total length of all lists destroyed  throughout the whole scan is $\Oh_{d,p}(n)$. It follows that the total running time is $\Oh_{d,p}(n)$, as promised.
\end{proof}

\paragraph*{Specifying a graph through its contraction sequence.} In all algorithmic statements we will assume that a graph is given by specifying its contraction sequence together with some auxiliary information encoding the edge relation. We now make this precise.

Let $\Pp_1,\ldots,\Pp_n$ be a contraction sequence of a graph $G$. We assume that every part participating in the partitions $\Pp_1,\ldots,\Pp_n$ 
(that is, every element of the union $\Pp_1\cup\cdots\cup\Pp_n$,
where each $\Pp_i$ is viewed as a set of sets of vertices)
is specified through a unique identifier taking a single machine word. For $\Pp_1$, the identifiers of (singleton) parts coincide with identifiers of the corresponding vertices. Then sequence $\Pp_1,\ldots,\Pp_n$ is represented by providing the following information for every time $s\in [n]$, $s>1$:
\begin{itemize} 
 \item The identifiers of the two parts $A,A'\in \Pp_{s-1}$ that get contracted at time $s$, and the identifier of the obtained part $B=A\cup A'\in \Pp_s$.
 \item A list of identifiers of parts $C\in \Pp_s$ such that the pair $B,C$ is impure in $G$ (that is, the impurities incident to $B$ in $\ImpGraph(G/\Pp_s)$).
 \item For each part $C$ on the list above, the relation (completeness, anti-completeness, or impurity) between $C$ and $A$ and between $C$ and $B$ in $G/\Pp_{s-1}$.
\end{itemize}
It is easy to see that this representation uniquely defines the graph $G$. Since the representation takes $\Oh_d(1)$ machine words at  any time $s$, we can thus represent an $n$-vertex graph of twin-width $d$ using $\Oh_d(n)$ machine words.

We now show that, for a graph given through a contraction sequence, one can reindex the vertex set using integers from $[n]$ so that the contraction sequence becomes convex.

\begin{lemma}\label{lem:reindexing}
 Suppose a graph $G$ is given by specifying a contraction sequence $\Pp_1,\ldots,\Pp_n$ of width $d$. Then one can in time $\Oh_d(n)$ compute a bijection $\eta\colon V(G)\to [n]$ such that mapping $G$ and $\Pp_1,\ldots,\Pp_n$ through $\eta$ yields an isomorphic graph $G'$ on vertex set $[n]$ and its contraction sequence $\Pp'_1,\ldots,\Pp'_n$ such that $\Pp'_1,\ldots,\Pp'_n$ is convex in the natural order on integers in $[n]$.
\end{lemma}
\begin{proof}
 Let $T$ be the rooted tree 
whose nodes are pairs $(P,s)$, where $s\in[n]$ and $P$ is a part in $\Pp_s$,
where $(P,s)$ is an ancestor of $(Q,t)$ if $t\le s$ and $Q\subset P$.
  It is easy compute $T$ in time $\Oh_d(n)$. Now, it remains to observe that as $\eta$ we can take the indexing of the leaves of $T$ (which naturally correspond to vertices of $G$) according to any pre-order of $T$.
\end{proof}

Note that if a graph is reindexed using Lemma~\ref{lem:reindexing}, then every part participating in the resulting contraction sequence is an interval in $[n]$. Hence, as the identifier of a part we can simply use a pair of vertices --- the left endpoint and the right endpoint --- and such identifiers can be computed in time $\Oh_d(n)$ by scanning the contraction sequence. We will therefore assume that contraction sequences 
are convex with respect to a fixed ordering of the vertices, and the (convex) parts are identified by their endpoints.

\paragraph*{First-order logic.}
We fix a countable set of variables, together with its enumeration. 
 If $\Omega$ is a set and $\tup x$ is a finite set of variables, then an {\em{$\tup x$-tuple}} with entries in $\Omega$ is a function from $\tup x$ to $\Omega$. Tuples are by convention denoted with boldface small letters, e.g. $\tup u$ or $\tup v$. The set of all $\tup x$-tuples with entries in $\Omega$ is denoted by $\Omega^{\tup x}$.
When $\tup a\in \Omega^{\tup x}$ is an $\tup x$-tuple 
and $b\in \Omega$, then by $\tup ab$ we denote the $(\tup x\cup\set{y})$-tuple
that extends $\tup a$ and maps the 
first  variable (according to the fixed enumeration of all variables)
$y$ not in $\tup x$, to $b$.

We consider standard first-order logic on graphs by modeling them as relational structures where the universe is the vertex set and there is a single binary predicate signifying adjacency. For a graph $G$, a formula $\varphi(\tup x)$, where $\tup x$ is the set of free variables of $\varphi$, and a tuple of vertices $\tup w\in V(G)^{\tup x}$, we write $G\models \varphi(\tup w)$, or $G,\tup w\models\varphi(\tup x)$, to denote that $\tup w$ satisfies $\varphi(\tup x)$ in $G$. We sometimes consider formulas with an explicitly partitioned set of free variables, e.g., $\varphi(\tup x,\tup y)$. 
Sentences are formulas with no free variables.

%

\paragraph*{Logical types.}
%
%
%
While the usual definition of a logical type of quantifier rank $k$ 
 of a tuple $\tup a$ of vertices of $G$ is the set of all formulas $\phi(\tup x)$ such that $G \models \phi(\tup a)$, we will rely on a definition which is more suitable for our purposes and is well known to be equivalent, by the result of Ehrenfeucht and Fra\"iss\'e (see for example~\cite{fmt_ebbinghaus_flum}).

Let $\tup x$ be a finite set of variables. An \emph{atomic type with variables $\tup x$} is a maximal consistent set $S$ of formulas of the form $x=y$, $x \not= y$, $E(x,y)$, $\lnot E(x,y)$, where $x, y \in \tup x$. Here by \emph{consistent} we mean that there is some graph $G$ and a tuple $\tup w\in V(G)^{\tup x}$ that satisfies all formulas occurring in the atomic type (this is decidable, as it is sufficient to consider graphs $G$ with $|G|\le |\tup x|$).

For $\tup a \in V(G)^{\tup x}$, the  \emph{atomic type of $\tup a$ in $G$}  is the atomic type with variables $\tup x$ which consists of all formulas of the form $E(x,y)$ or $x=y$, where $x,y\in \tup x$, such that $G,\tup a\models x=y$ or $G,\tup a\models E(x,y)$.

\begin{definition}
Let $G$ be a graph $\tup x$ a finite set of variables and $k \in \N$. For every $\tup a \in V(G)^{\tup x}$ we define its \emph{type of quantifier rank $k$}, denoted $\tp^k(\tup a)$, as follows.
\begin{itemize} 
\item If $k=0$, then $\tp^0(\tup a)$ is the atomic type of $\tup a$ in $G$.

\item If $k>0$, then $\tp^k(\tup a) = \setof{\tp^{k-1}(\tup ab)}{b\in V(G)}$.
\end{itemize}
For $k\ge 1$ we also set $\tp^k(G) = \setof{\tp^{k-1}(a)}{a \in V(G)}$.
\end{definition}
This definition is usually intuitively explained in terms of Ehrenfeucht-Fra\"iss\'e games. Namely, two $\tup x$-tuples $\tup a$ and $\tup b$ of vertices
of two graphs $G$ and $H$, respectively,
have equal types of quantifier rank $k$ if and only if Duplicator 
wins the $k$-round game on the graphs $G$ and $H$, 
where the initial pebbles in $G$ and $H$ are placed on the vertices occurring in $\tup a$ and in $\tup b$, respectively.
Indeed, suppose $\tp^k(\tup a)=\tp^k(\tup b)$, where $k>0$, and that Spoiler places a pebble on a vertex $c$ of $G$.
Then, since $\tp^{k-1}(\tup ac)\in \tp^{k}(\tup a)$ 
by definition and $\tp^k(\tup a)=\tp^k(\tup b)$, 
we have that there is some $d\in\tp^k(\tup b)$ 
such that $\tp^{k-1}(\tup bd)\in \tp^k(\tup b)$.
Then Duplicator responds by placing the pebble on the vertex $d$,
and we have that $\tp^{k-1}(\tup ac)=\tp^{k-1}(\tup bd)$ so,
by inductive assumption, Duplicator wins in the $k-1$ round game 
from the current configuration, which shows that Duplicator has a winning strategy in the $k$ round game starting from $\tup a$ and $\tup b$. The implication in the other direction proceeds similarly.

\medskip
As is well known, the set of types of $\tup x$-tuples of quantifier rank $k$ that are realized by some tuple $\tup a$, in some graph, is non-computable, even though this set has size bounded in terms of $\tup x$ and $k$. To overcome this problem,
the usual solution is to define the set of abstract types 
(that may not be realized as actual types), which is computable from $\tup x$ and $k$, has bounded size, and contains all types that may arise. This is done as follows.

Define $\Types_{\tup x}^0$ as the set of all atomic types over $\tup x$ and $\Types_{\tup x}^k:=\setof{M}{M \subset \Types_{\tup xy}^{k-1}}$.
Note that for any $G$ and any $\tup a \in V(G)^{\tup x}$ it holds that   $\tp^k(\tup a) \in \Types_{\tup x}^k$, but $ \Types_{\tup x}^k$ can also contain objects which are not realized by any tuple of vertices $\tup a$ of any graph.

For a graph $G$ we set $\Types_{\tup x}^k(G) := \setof{\tp^k(\tup a)}{\tup a \in V(G)^{\tup x}}$. Note that we have $\tp^k(G) = \Types_{x}^{k-1}(G)$.

The following is well known and follows from the fact that our definition of types is equivalent to the usual definition of types using formulas.
\begin{proposition}\label{prop:types_mc}
Let $G$ be a graph, $\tup x$ a set of variables and $k \in \N$.
\begin{itemize} 
\item $|\Types_{\tup x}^k(G)| = \Oh_{k,\tup x}(1)$,
\item For any $\tup a \in V(G)^{\tup x}$ and any first-order formula $\phi(\tup x)$ of quantifier rank at most $k$ one can determine whether $G \models \phi(\tup a)$ from $\tp^k(\tup a)$ in time $\Oh_{k,\tup x}(1)$.
\item For any first-order sentence $\phi$ of quantifier rank at most $k$ one can determine whether $G \models \phi$ from $\tp^k(G)$ in time $\Oh_{k}(1)$.
\end{itemize} 
\end{proposition}

\section{Local types}\label{sec:types}

In this section we define local types of quantifier rank $k$ for partitioned graphs, or \emph{local $k$-types} for short. They provide a framework for the results proved in the rest of the paper. The key lemmas are Lemma~\ref{lem:ltp_compositionality} and Lemma~\ref{lem:ltp_consistency} and their corollaries Lemma~\ref{lem:tp-merge} and Lemma~\ref{lem:tp-warp}.

\subsection{Local types for partitioned graphs}
Let $G$ be a graph and $\Pp$ be a partition of its vertex set,
and let $\tup x$ be a set of variables. For an $\tup x$-tuple $\tup a\in V(G)$ write 
$\iwarp {\tup a}{\Pp}$ for the $\tup x$-tuple 
$\tup u\in \Pp^{\tup x}$ such that $\tup u(x)$ is the part containing $\tup a(x)$, for all $x\in \tup x$.

\begin{definition}\label{def:ltp}
    Let $G$ be a graph, $\Pp$ be a partition of its vertex set,  $\tup x$ a nonempty set of variables, and $k \in \N$.
For any $\tup a \in V(G)^{\tup x}$ we define the \emph{local $k$-type of $\tup a$}, denoted $\ltp_{\Pp}^k(\tup a)$, as follows:
\begin{itemize} 
\item  $\ltp_{\Pp}^0(\tup a)$ is the atomic type of $\tup a$ together with the $\tup x$-tuple $\iwarp{\tup a}{\Pp}\in \cal P^{\tup x}$ of parts of $\cal P$ corresponding to $\tup a$,
\item for $k>0$, let $\ltp_{\Pp}^k(\tup a) = \setof{\ltp_{\Pp}^{k-1}(\tup ab)}{b \in w \text{ for some $w\in \Pp$ with $\dist_{\Pp}(\iwarp{\tup a}{\Pp}, w)\le 2^{k-1}$}}$.
\end{itemize}
\end{definition}

As with usual types of quantifier-rank $k$ defined in the previous section, it is often convenient to think about equality of local types in terms of games.
We now briefly describe the corresponding variant of Ehrenfeucht-Fra\"iss\'e games game.
This game will be played on a single graph $G$ with a fixed partition $\Pp$ of its vertex set (one can also imagine it being played on two copies of the same graph with the same partition).
The starting position of the game is determined by two $\tup x$-tuples $\tup a$ and $\tup b$  of vertices of $G$ (where $\tup x$ is nonempty) such that for every $y \in \tup x$ we have that $\tup a(y)$ is in the same part of $\Pp$ as $\tup b(y)$. The game is played for $k$ rounds as the usual Ehrenfeucht-Fra\"iss\'e game  with the following extra restrictions on the moves of the players: (1) In the $i$th round, Spoiler picks one of the tuples $\tup a$ and $\tup b$, and he will then proceed to extending it.
Suppose that he picks $\tup a$, the other case being symmetric.
Spoiler then picks a vertex $a$ in any part $P \in \Pp$ such that $\dist_{\Pp}(P,Q)\le 2^{k-1}$, where $Q$ is some part containing a vertex of $\tup a$. He then appends $a$ to $\tup a$ to form the tuple $\tup aa$. (2) Duplicator replies by picking a vertex $b$ in the same part $P$ and extending the other tuple $\tup b$ to $\tup bb$. 
 The game then continues to the next round, with $\tup aa$ and $\tup bb$ forming the new position.
 Duplicator wins after $k$ rounds if the two tuples have equal atomic types. 

It is not difficult to see that Duplicator wins the $k$-round game described above, starting from the configuration $\tup a$ and $\tup b$, if and only if $\ltp^k_{\Pp}(\tup a)=\ltp^k_{\Pp}(\tup b)$.
This is made formal in the following proposition, whose proof is an immediate consequence of Definition~\ref{def:ltp}.

\begin{proposition}Let $G$ be a graph and $\Pp$ be a partition of the vertices of $G$, and let $\tup x$ a tuple of variables and $k \in \N$. Then the following holds for any $\tup a, \tup b \in V(G)^{\tup x}$:
\begin{itemize} 
\item  $\ltp_\Pp^0(\tup a) = \ltp_\Pp^0(\tup b)$ if and only if the atomic types of $\tup a$ and $\tup b$ are the same, and 
 $\iwarp{\tup a}{\Pp}=\iwarp{\tup b}{\Pp}$;
\item If $k > 0$ then $\ltp_\Pp^k(\tup a) = \ltp_\Pp^k(\tup b)$ if and only if for any  $P \in \Pp$ with $\dist_\Pp(\iwarp{\tup a}{\Pp}, P)\le 2^{k-1}$ the following holds:
for any $c \in P$ there exists $c' \in P$ such that $\ltp_\Pp^{k-1}(\tup ac)  = \ltp_\Pp^{k-1}(\tup {b}c')$, and conversely,
for any $c' \in P$ there exists $c \in P$ such that $\ltp_\Pp^{k-1}(\tup ac)  = \ltp_\Pp^{k-1}(\tup {b}c')$.
\end{itemize}
%
\end{proposition}

We will also need to have an abstract set containing all types which could potentially occur for any $k\in \N$ and $\tup u \in \Pp^{\tup x}$. Note that this includes also types which are not realized in $G$ (or even in any graph).
\begin{definition}\label{def:Types}
Let $\tup x$ be  a nonempty set of variables and $k \in \N$.		
Fix a graph $G$ together with a vertex-partition $\Pp$.
For $\tup u \in \Pp^{\tup x}$ we define $\Types_{\tup u, \Pp}^0:=\setof{(\alpha,\tup u)}{\alpha \text{ is an atomic type with variables $\tup x$}}$. For $k > 0$ let $M$ be the set of all parts $w$ of $\Pp$ with $\dist_{\Pp}(\tup u,w)\le 2^{k-1}$ and let $M':=\bigcup_{w \in M} \Types_{\tup uw,\Pp}^{k-1}$. We then define
$$\Types_{\tup u, \Pp}^k:=\setof{S}{S\subset M'}.$$
Define also $\Types_{\tup u, \Pp}^k(G):=\setof{\ltp^{k}(\tup a)}{\tup a \in V(G)^{\tup x}, \tup u = \iwarp{\tup a}{\Pp}}$.
\end{definition}
Then $\Types_{\tup u, \Pp}(G)$ is the set of all local $k$-types realized in $\tup u$, and is a subset of $\Types^k_{\tup u,\Pp}$.

\subsection{Properties of local types}
In this section we establish the properties of local $k$-types used in the rest of the paper.

In the rest of this paper, we assume that we have fixed a graph $G$ and a contraction sequence $\cal P_1,\ldots, \cal P_n$ of $G$.
We write $\ltp_s^k(\cdot )$ to denote $\ltp_{\Pp_s}^k(\cdot )$,
$\Types_{\tup u,s}(\cdot )$ to denote $\Types_{\tup u,\Pp}(\cdot )$,
and $\dist_s(\cdot,\cdot)$ to denote $\dist_{\cal P_s}(\cdot ,\cdot)$.

The following two lemmas establish some basic properties of local $k$-types.

\begin{lemma}\label{lem:ltp_basic}
The following holds at any time $s\in [n]$ and $k \ge 1$.
\begin{itemize} 
\item If $\ltp_s^k({\tup a}) = \ltp_s^k({\tup b})$, then $\ltp_s^{k-1}({\tup a}) = \ltp_s^{k-1}({\tup b})$.
\item If $\ltp_s^k({\tup a}) = \ltp_s^k({\tup b})$, then $\iwarp{\tup a}{s} = \iwarp{\tup b}{s}$.
\end{itemize}
\end{lemma}
\begin{proof}
Follows immediately from the definition.
\end{proof}

\begin{lemma}\label{lem:ltp_number_of_types}
Let $s \in [n]$ be a time and $\tup x$ a tuple of variables, and $k\ge 0$. 
Then $|\Types^k_{\tup u,s}|\leq \Oh_{d,k,\tup x}(1)$, for all $\tup u \in \Pp_s^{\tup x}$. Moreover, given vicinity $\Vic_s^{2^k}(\tup u)$, one can compute $\Types^k_{\tup u,s}$ in time $\Oh_{d,k,\tup x}(1)$.
\end{lemma}
\begin{proof}
We bound the size of $\Types^k_{\tup u,s}$ by induction on $k$. For $k=0$, there is only a bounded number of atomic types of $\tup x$-tuples. For $k>0$, the number of parts $w \in \Pp_s$ with $\dist_s(\tup u, w) \le 2^{k-1}$ is at most $|\tup x|(d^k+1)$, let us denote them by $w_1, \ldots, w_m$. For each such part $w_i$ we know from the induction hypothesis that $|\Types^{k-1}_{\tup uw_i,s}|$ is bounded in terms of $d, k-1$ and $|\tup x| + 1$. Since each member of $\Types^k_{\tup u, s}$ is a subset of $ \bigcup_{1 \le i \le m}\Types^{k-1}_{\tup uw_i,s}$, we have that $|\Types^{k}_{\tup u, s}| \le 2^T$, where $T= \Sigma_{1 \le i \le m}|\Types^{k-1}_{\tup uw_i,s}|$, and the result follows.

The part about computation of $\Types^k_{\tup u,s}$  follows from the bound on the size of $\Types^k_{\tup u,s}$ and the fact that every part $w \in \Pp_s$ involved in the definition of $\Types^k_{\tup u,s}$ is in $\Vic_s^{2^k}(\tup u)$.
\end{proof}

The following lemma relates local types for partitioned graphs to usual first-order types, as defined in the preliminaries.
\begin{lemma}\label{lem:ltp_type_as_local_type}
Let $\tup x$ be a tuple of variables and $\tup a \in V(G)^{\tup x}$.
One can compute $\tp^k(\tup a)$ from $\ltp_n^k(\tup a)$ in 
time~$\Oh_{k,d,\tup x}(1)$.
\end{lemma}
\begin{proof}
At the time $n$ there is only one part $w$ in $\Pp_n$, and so all the information related to this part and distances which is contained in $\ltp_n^k(\tup a)$ is trivial, and after removing it all that is left is the definition of $\tp^k(\tup a)$.
One can do this removal of information by recursing on members of $\ltp_n^k(\tup a)$ (which are themselves local $(k-1)$-types), and to bound the runtime it suffices to show that the size of $\ltp_n^k(\tup a)$ is bounded by $\Oh_{d,k,\tup x}(1)$.
Let $\tup w$ be the only member of  $\Pp_n^{\tup x}$.
Since $\ltp_n^k(\tup a)$ consists of members of $\Types_{\tup ww, s}^{k-1}$ and by Lemma~\ref{lem:ltp_number_of_types} we have $|\Types_{\tup ww, s}^{k-1}| \leq \Oh_{d,k,\tup x}(1)$, it follows by an easy inductive argument that $\ltp_n^k(\tup a)$ is of size bounded by $\Oh_{d,k,\tup x}(1)$, and the result follows.
\end{proof}

%



The next lemma is a version of compositionality of local types and plays a key role in computing local types.
\begin{lemma}\label{lem:ltp_compositionality}
Fix two disjoint sets of variables $\tup x$ and $\tup y$. Let $\tup a, \tup a' \in V^{\tup x}$ and $\tup b, \tup b' \in V^{\tup y}$ be such that $\ltp_s^k(\tup a)=\ltp_s^k(\tup a')$ and $\ltp_s^k(\tup b)=\ltp_s^k(\tup b')$. Let $\tup u=\iwarp{\tup a}{s}$ and $\tup v=\iwarp{\tup b}{s}$ and assume that $\dist_s(\tup u, \tup v)> 2^k$. Then $\ltp_s^k(\tup{ab})=\ltp_s^k(\tup {a'b'})$.
\end{lemma}

\begin{proof}
We prove the statement by induction on $k$. For $k=0$, to prove that $\ltp_s^0(\tup{ab})=\ltp_s^0(\tup {a'b'})$, we have to show that the atomic types of $\tup{ab}$ and $\tup{a'b'}$ are the same. Fix an atomic formula $\phi(x,y)$, with $x,y\in\tup x\cup\tup y$. We show that $G,\tup a\tup b\models \phi(x,y)$ if and only if $G,\tup a'\tup b'\models \phi(x,y)$.
If $x$ and $y$ both belong to $\tup x$ then the conclusion follows by assumption that $\ltp^0_s(\tup a)=\ltp^0_s(\tup a')$.
 The same holds if $x$ and $y$ both belong to $\tup y$.
 
 So, by symmetry, it is enough to consider the case when $x\in \tup x$ and $y\in\tup y$. Since by our assumption  $\dist_s(\tup u, \tup v)> 2^0=1$, any part of $\tup u$ is pure to any part $\tup v$, and so in particular the part $\iwarp{\tup a(x)}{s}$ is pure towards  $\iwarp{\tup b(y)}{s}$. Because $\tup a(x),\tup a'(x) \in \iwarp{\tup a(x)}{s}$ and $\tup b(y),\tup b'(y) \in \iwarp{\tup b(y)}{s}$, this implies that  $G,\tup a\tup b\models\phi(x,y)$ if and only if $G,\tup a'\tup b'\models\phi(x,y)$, as required.

For $k>0$, let $c$ be a vertex in a part $w=\iwarp{c}{s}$ such that $\dist_s(\tup{uv}, w) \le 2^{k-1}$. Our task is to show that there exists $c' \in w$ such that $\ltp_s^{k-1}(\tup{ab}c)=\ltp_s^{k-1}(\tup {a'b'}c')$. Since $\dist_s(\tup u, \tup v)> 2^k$, exactly one of $\dist_s(\tup{u}, w) \le 2^{k-1}$ and $\dist_s(\tup{v}, w) \le 2^{k-1}$ has to hold; without loss of generality assume that $\dist_s(\tup{u}, w) \le 2^{k-1}$ holds. Since $\ltp_s^k(\tup a)=\ltp_s^k(\tup a')$, there exists $c' \in w$ such that $\ltp_s^{k-1}(\tup{a}c)=\ltp_s^{k-1}(\tup {a'}c')$. Because $\dist_s(\tup u, \tup v)> 2^k$ and $\dist_s(\tup{u}, w) \le 2^{k-1}$, we have $\dist_s(\tup uw, \tup v)> 2^{k-1}$, so we can apply the induction hypothesis to $\tup ac$, $\tup a'c'$ and $\tup b$, $\tup b'$, which yields that  $\ltp_s^{k-1}(\tup{ab}c)=\ltp_s^{k-1}(\tup {a'b'}c')$, as desired.
\end{proof}

The following lemma follows directly from Lemma~\ref{lem:ltp_compositionality}, except for the part about efficient computation.
\begin{lemma}\label{lem:tp-merge}
 Let $s\in [n]$ be a time and 
 $\tup x$ and $\tup y$ are disjoint sets of variables. Suppose $\tup u\in \Pp_s^{\tup x}$ and $\tup v\in \Pp_s^{\tup y}$ are tuples of parts such that
 $\dist_s(\tup u,\tup v)>2^k.$
 Then there is a function $f\colon \Types^k_{\tup u,s}\times \Types^k_{\tup v,s}\to \Types^k_{\tup u\tup v,s}$ such that for every pair of tuples $\tup a\in V^{\tup x}$ and $\tup b\in V^{\tup y}$ satisfying $\tup u=\iwarp{\tup a}{s}$ and $\tup v=\iwarp{\tup b}{s}$, we have
 $$\ltp_s^k(\tup a\tup b) = f(\ltp_s^k(\tup a),\ltp_s^k(\tup b)).$$
 Moreover, given $k$, $\tup u$, $\tup v$, and the vicinity $\Vic_s^{2^k}(\tup u\tup v)$, one can compute $f$ in time $\Oh_{d,k,\tup x,\tup y}(1)$. 
\end{lemma}
Regarding the computation of function $f$ in the above lemma, by ``computing $f$ in time $\Oh_{d,k,\tup x,\tup y}(1)$'' we do not mean just evaluating $f$ on any given input in desired time, but constructing the whole input-output table for $f$. The reason why this can be computed from 
$k$, $\tup u$, $\tup v$ and $\Vic_s^{2^k}(\tup u\tup v)$ 
in time $\Oh_{d,k,\tup x,\tup y}(1)$ is that the input and output sets have size bounded by $\Oh_{d,k,\tup x,\tup y}(1)$
and the proof in Lemma~\ref{lem:ltp_compositionality} uses only information from $\Vic_s^{2^k}(\tup u\tup v)$. A concrete approach to implementing this computation, similar to that presented in~\cite{tww1}, can be found in Appendix~\ref{app:effective}.


The next lemma will allow us to determine how the $k$-type of a tuple $\tup a$ develops over time. 
\begin{lemma}\label{lem:ltp_consistency}
Let $s \in [n]$ be a time and let $\tup a \in V^{\tup x}$,  $\tup a' \in V^{\tup x}$ be two tuples of vertices such that $\ltp^k_s(\tup a) = \ltp_s^k(\tup a')$. Then $\ltp^k_{s+1}(\tup a) = \ltp_{s+1}^k(\tup a')$.
\end{lemma}

\begin{proof}
By induction on $k$. For $k=0$ note that $\ltp_s^0(\tup a) = \ltp_s^k(\tup a')$ implies that atomic types of $\tup a$ and $\tup a'$ are the same and $\iwarp{\tup a}{s} = \iwarp{\tup a'}{s}$. It is easily seen that then also $\iwarp{\tup a}{s+1} = \iwarp{\tup a'}{s+1}$, as desired.

For $k > 0$, let $\tup u = \iwarp{\tup a}{s+1} = \iwarp{\tup a'}{s+1}$. We need to show that for any $w \in \cal \Pp_{s+1}$ with $\dist_{s+1}(\tup u, w) \le 2^{k-1}$ and any $b \in w$ there is $b' \in w$ such that $\ltp^{k-1}_{s+1}(\tup ab) = \ltp^{k-1}_{s+1}(\tup a'b')$, and symmetrically, that for any $b' \in w$ there is $b \in w$ such that $\ltp^{k-1}_{s+1}(\tup ab) = \ltp^{k-1}_{s+1}(\tup a'b')$. We focus on the first option; the proof of the second one is analogous. Let $v = \iwarp{b}{s}$. We distinguish two possibilities:
\begin{itemize} 
\item $\dist_s(\tup u, v) \le 2^{k-1}$: In this case, since $\ltp_s^k(\tup a) = \ltp_s^k(\tup a')$, there exists $b' \in v$ such that $\ltp_s^{k-1}(\tup ab) = \ltp_s^{k-1}(\tup a'b')$. Then by induction hypothesis it follows that $\ltp^{k-1}_{s+1}(\tup ab) = \ltp^{k-1}_{s+1}(\tup a'b')$, as desired.
\item $\dist_s(\tup u, v) > 2^{k-1}$: In this case we note that $\ltp^k_s(\tup a) = \ltp_s^k(\tup a')$ implies that $\ltp^{k-1}_s(\tup a) = \ltp_s^{k-1}(\tup a')$, and we set $b':=b$. We can now apply Lemma~\ref{lem:ltp_compositionality} to $\tup a$, $\tup a'$ and $b$, $b'$ to see that $\ltp^{k-1}_{s}(\tup ab) = \ltp^{k-1}_{s}(\tup a'b')$, and by induction hypothesis it follows that $\ltp^{k-1}_{s}(\tup ab) = \ltp^{k-1}_{s}(\tup a'b')$, as desired.
\end{itemize}
\end{proof}


Lemma~\ref{lem:ltp_consistency} implies that there exists a function which maps $\ltp_s^k(\tup a)$ to $\ltp_{s+1}^k(\tup a)$, and by induction we get the following lemma.
\begin{lemma}\label{lem:tp-warp}
 Let $s,t\in [n]$ be times with $s\leq t$. Suppose $\tup u\in \Pp_s^{\tup x}$ and let $\tup v=\warp{\tup u}{s}{t}$. Then there exists a function $f\colon \Types^k_{\tup u,s}\to \Types^k_{\tup v,t}$ such that for every tuple $\tup a\in V^{\tup x}$ satisfying $\tup u=\iwarp{\tup a}{s}$, we have
 $$\ltp_t^k(\tup a) = f(\ltp_s^k(\tup a)).$$
Moreover, if $t=s+1$, then given $k$, $\tup u$, $\tup v$ and the relevant region $\Rel_s^{2^k(|\tup x| + 1)}$, one can compute $f$ in time $\Oh_{d,k,\tup x}(1)$, provided that for every $y \in \tup x$ we have that $\tup v(y) \in \Rel_s^{2^k|\tup x|}$.
\end{lemma}

As in the case of Lemma~\ref{lem:tp-merge}, the whole input-output table of function $f$ can be computed in time $\Oh_{d,k,\tup x,\tup y}(1)$ from 
$k$, $\tup u$, $\tup v$ and $\Rel_s^{2^k}$, since the proof of Lemma~\ref{lem:ltp_consistency} uses only information from $\Rel_s^{2^k(|\tup x| + 1)}$. Again, a concrete approach to implementing this computation, similar to that presented in~\cite{tww1}, can be found in Appendix~\ref{app:effective}.

We will also use the fact that when going from time $s$ to $s+1$ the local $k$-types of tuples in parts which are not in the trigraph $\Rel_s^{2^k}$ are not affected.
\begin{lemma}\label{lem:relevant}
Let $\tup x$ be a finite set of variables, $s \in [n]$ a time, $k\in N$ and let $\tup u \in \Pp_s^\tup{x}$ be such for every $y \in \tup x$ it holds that $\tup u(y) \not\in V(\Rel_s^{2^k})$. Then $\warp{\tup u}{s}{s+1} = \tup u$, and
for every  $\tup a$ with $\tup u = \iwarp{\tup a}{s}$ we have that $\ltp_s^{k}(\tup a) = \ltp_{s+1}^{k}(\tup a)$.
In particular, $\Types_{\tup u,s+1}^k(G) = \Types_{\tup u,s}^k(G)$.
\end{lemma}
\begin{proof}
To prove the first part, note that the definition of $\Rel_s^{2^k}$ says that both parts of $\Pp_s$ which are merged into a part of $\Pp_{s+1}$ (these are the only parts of $\Pp_s$ which change when going to time $s+1$) are included in $V(\Rel_s^{2^k})$. Since for every $y\in \tup x$  we have by our assumption that $\tup u(y) \not\in V(\Rel_s^{2^k})$, $\tup u(y)$ is the same in time $s+1$ as it was in $s$.

We prove  that for every  $\tup a$ with $\tup u = \iwarp{\tup a}{s}$ we have that $\ltp_s^{k}(\tup a) = \ltp_{s+1}^{k}(\tup a)$  by induction on $k$.  For $k=0$, let $\tup a$ be an arbitrary tuple with $\tup u = \iwarp{\tup a}{s}$. 
Then $\ltp_s^{k}(\tup a) = (S,\tup u)$, where $S$ is the atomic type of $\tup a$ in $G$.  In time $s+1$ the atomic type of $\tup a$ clearly stays the same, and since  $\warp{\tup u}{s}{s+1} = \tup u$, we have that $\ltp_{s+1}^{k}(\tup a) = (S,\tup u)$.

Let $k>0$. In this case every member of $\ltp_s^{k}(\tup a)$ is of the form $\ltp_s^{k-1}(\tup ab)$ for some $b \in w$, where $\dist_s(\tup u,w) \le 2^{k-1}$. Then we have that $w \not \in V(\Rel_s^{2^{k-1}})$, and so $\tup uw \not \in  V(\Rel_s^{2^{k-1}})$. Then by induction hypothesis we get that $\ltp_s^{k-1}(\tup ab) = \ltp_{s+1}^{k-1}(\tup ab)$, and since $\dist_{s+1}(\tup u,w) \le 2^{k-1}$, we have that $\ltp_{s+1}^{k-1}(\tup ab) \in  \ltp_{s+1}^{k}(\tup a)$ and so $\ltp_{s}^{k-1}(\tup ab) \in  \ltp_{s+1}^{k}(\tup a)$, as desired.

For the other direction, every member of  $\ltp_{s+1}^{k}(\tup a)$ is of the form $\ltp_{s+1}^{k-1}(\tup ab)$ for some $b \in w$, where $\dist_{s+1}(\tup u,w) \le 2^{k-1}$. From the definition of $\Rel_s^{2^k}$ and the assumption that no part from $\tup u$ is in $\Rel_s^{2^k}$  it then follows that $\dist_{s+1}(v,w)> 2^{k-1}$, where $v$ is the  part of $\Pp_{s+1}$ which was obtained by contracting two parts of $\Pp_s$. This means that $w \not \in V(\Rel_s^{k})$, and so we can apply induction hypothesis to claim that  $\ltp_{s}^{k-1}(\tup ab)  = \ltp_{s+1}^{k-1}(\tup ab) $, and so $\ltp_{s+1}^{k-1}(\tup ab) \in \ltp_{s}^{k}(\tup a)$, as desired.
\end{proof}

\subsection{Reproving results of Bonnet et al.~\cite{tww1}}

With the machinery from the previous subsection we can now reprove the two fundamental results about graphs of bounded twin-width, proved by Bonnet et al. in~\cite{tww1}.

\paragraph{Fixed-parameter tractable model checking.}

The first results concerns linear-time fixed-parameter tractability of model-checking first-order logic on graphs of bounded twin-width.

\begin{theorem}\label{thm:mc}
Let $G$ be a graph on $n$ vertices represented through its contraction sequence $\Pp_1,\ldots, \Pp_n$ of width $d$.
Then for any sentence $\phi$ one can decide whether $G \models \phi$ in time $\Oh_{d,\varphi}(n)$.
\end{theorem}

\begin{proof}
Let $q$ be the quantifier rank of $\phi$ and set $k\coloneqq q-1$ and $r\coloneqq 2^k$. We will show how to compute the set $\Types^{k}_x(G)$ in desired time, and since $\tp^q(G) = \Types^{q-1}_x(G)$, the result will follow by  Proposition~\ref{prop:types_mc}.

As a preprocessing step, the algorithm computes in time $\Oh_{d,k}(n)$ the trigraphs $\Rel_s^{r}$ for all $s\in [n-1]$; this can be done by Lemma~\ref{lem:compute-affected}.
For the rest of the proof, let us for any time $s \in [n]$ denote by $T_s$ the set of all sets of realized types at time $s$, i.e. $T_s \coloneqq \setof{\Types_{w,s}^k(G)}{w \in \Pp_s}$.

The algorithm first computes $T_1$ by computing $\Types_{w,1}^k(G)$ for each $w \in \Pp_1$; since each such part $w$ contains exactly one vertex, this can be done in time $\Oh_{k}(1)$ for any $w$, and so this takes time $\Oh_{k}(n)$ in total. From this point on the algorithm will proceed through times $2$ to $n$ and for every time $s$ it will compute $T_{s}$ from $T_{s-1}$.
By Lemma~\ref{lem:relevant}, any part $w$ of $\Pp_{s-1}$ which is not in $\Rel_{s-1}^{2^k}$ is the same in $\Pp_s$ as in $\Pp_{s-1}$ and we have that $\Types_{w,s-1}^k(G) = \Types_{w,s}^k(G)$, which means that the computation only needs to be performed on parts from $\Rel_{s-1}^{2^k}$.
We distinguish the following two possibilities:
\begin{itemize} 
\item If $v,w$ are the two parts of $\Pp_{s-1}$ which get contracted into a part $u \in \Pp_s$, then the algorithm applies the function from Lemma~\ref{lem:tp-warp} to all members of $\Types_{v,s-1}^k(G)$ and $\Types_{w,s-1}^k(G)$ and collects the results into $\Types_{u,s}^k(G)$.
\item If $w$ is any other part in $\Rel_{s-1}^{r}$, then the algorithm applies the function from Lemma~\ref{lem:tp-warp} to all members of $\Types_{w,s-1}^k$ and collects the results into $\Types_{w,s}^k(G)$.
\end{itemize}
In each of the above cases the computation can be done in time $\Oh_{d,k}(1)$, since each application of the function from Lemma~\ref{lem:tp-warp} can be done in time $\Oh_{d,k,1}(1)$ and by Lemma~\ref{lem:ltp_number_of_types} we have that $|\Types_{w,{s-1}}^k(G)| \leq \Oh_{d,k}(1)$. Moreover, since $|\Rel_{s-1}^{r}|\leq \Oh_{d,k}(1)$, the computation of $T_s$ from $T_{s-1}$ can be done in time $\Oh_{d,k}(1)$. There are $n-1$ steps to obtain $T_n$ and so the whole computation takes time $\Oh_{d,k}(n)$. Now $T_n$ contains only $\Types_{w,n}^k(G)$ where $w$ is the only part of~$\Pp_n$. By Lemma~\ref{lem:ltp_type_as_local_type}, from each local $k$-type in $\Types_{w,n}^k(G)$ one can compute the corresponding $k$-type from $\Types^k_{x}(G)$   in time $\Oh_{k}(1)$.
This finishes the proof.
\end{proof}

\paragraph{Stability under first-order interpretations.}

Another result of \cite{tww1}, which can be reproven with the machinery of local types, is that graph classes of bounded twin-width are closed under first-order interpretations.
For a graph $G$ and a first-order formula with two free variables $\phi(x, y)$ we define $\phi(G)$ to be the graph with vertex set $V(G)$ and edge set consisting of all the pairs $uv$ for which $G \models \phi(u, v) \land \phi(v, u)$.

\begin{theorem}
Let $G$ be a graph of twin-width at most $d$ and $\phi(x, y)$ be a first-order formula with two free variables.
Then the twin-width of the graph $\phi(G)$ is $\Oh_{d, \phi}(1)$.
\end{theorem}
\begin{proof}
Let $q$ be the quantifier rank of $\phi$.
Let $\Pp_1,\ldots, \Pp_n$ be a contraction sequence of $G$ of width at most $d$.
For every partition $\Pp_t$ with $1 \le t \le n$ we define its refinement $\Pp'_t$ by splitting every part of $\Pp_t$ with respect to local $q$-types.
Formally:
\[
\Pp'_t = \bigcup_{A \in \Pp_t} \setof{\setof{v \in A}{\ltp_{\Pp_t}^q(v) = S}}{S \in \Types^q_{A, \Pp_t}}.
\]
(Here, in the notation $\ltp_{\Pp_t}^q(v)$ we formally treat $v$ as a singleton tuple.)
Also, let $\Pp'_{n+1}$ be the coarsest partition of $V(G)$, the one with one part.

First, let us fix a part $A' \in \Pp'_t$ that originates from a part $A \in \Pp_t$ for some $1 \le t \le n$.
Our goal is to bound the number of parts in $\Pp'_t$ which are impure towards $A'$.
Consider any $B' \in \Pp'_t$ which is impure towards $A'$ and originates from a part $B \in \Pp_t$.
If we had $\dist_{\Pp_t}(A, B) > 2^q$, then by Lemma~\ref{lem:ltp_compositionality} for all pairs  $(a,b) \in A'\times B'$ the local type $\ltp_{\Pp_t}^q(ab)$ would be the same. Hence, by Lemma~\ref{lem:ltp_consistency}, the type $\tp^q(ab)$ would also be the same for all such pairs $(a,b)$, implying that $A'$ and $B'$ are a pure pair in $\phi(G)$.
So we have $\dist_{\Pp_t}(A, B) < 2^q$. Note that there are only $\Oh_{d, q}(1)$ such parts $B$ and,
by Lemma~\ref{lem:ltp_number_of_types}, each  of them is split into $\Oh_{d, q}(1)$ parts in $\Pp'_t$. So $A'$ is impure towards $\Oh_{d, q}(1)$ other parts of $\Pp'_t$.

Second, again by Lemma~\ref{lem:ltp_consistency}, for every $1 \le t \le s \le n + 1$ we have that $\Pp'_t$ is a refinement of $\Pp'_s$.
Observe also that for any $1 \le t \le n$, we can obtain $\Pp'_{t+1}$ by merging $\Oh_{d, q}(1)$ parts of~$\Pp'_t$.
Indeed, by Lemma~\ref{lem:relevant}, local $q$-types might change only for vertices in $\Rel_{t}^{2^q}$ and, by Lemma~\ref{lem:ltp_number_of_types}, there are at most $\Oh_{d, q}(1)$ different local $q$-types among them.
Similarly as in \cite[Lemma 8]{tww1}, we argue that we can extend the sequence $\Pp'_1, \ldots, \Pp'_{n+1}$ to a contraction sequence of $\phi(G)$ by contracting in any way parts of $\Pp'_t$ to obtain $\Pp'_{t+1}$. Clearly, the width of this contraction sequence is $\Oh_{d, q}(1)$.
\end{proof}

\newcommand{\Prx}{\mathbb{P}}
\newcommand{\Wrp}{\mathbb{W}}

\section{Query answering}\label{sec:queries}

In this section we prove Theorem~\ref{thm:intro-main-qa}.
For the remainder of this section let us fix a graph $G$ and a contraction sequence $\Pp_1,\ldots,\Pp_{n}$ of $G$ of width $d$, where $n=|V(G)|$. In all algorithmic statements that follow, we assume that $G$ and $\Pp$ are given on input.

Throughout this section our data structures work with the standard word RAM model.

\subsection{Proximity oracle}

For vertices $u,v\in V(G)$ and $r\in \N$, we define
$$\prx_r(u,v)=\min\{t~|~\dist_t(\iwarp{u}{t},\iwarp{v}{t})\leq r\}.$$
In other words, $\prx_r(u,v)$ is the first time $t$ such that the parts of $\Pp_t$ containing $u$ and $v$ are at distance at most $r$ in the impurity graph $\Imp_t$. Note that whenever $u\neq v$, we have $1<\prx_r(u,v)\leq n$. The main goal of this section is to construct an auxiliary data structure for answering queries about the values of $\prx_r(\cdot,\cdot)$. This is described in the lemma below.

\begin{lemma}\label{lem:proximity-ds}
 For a given $r\in \N$, one can in time $\Oh_{d,r}(n)$ compute a data structure that can answer the following queries in time $\Oh_{d,r}(\log \log n)$: given $u,v\in V(G)$, output $\prx_r(u,v)$.
\end{lemma}
By Lemma~\ref{lem:reindexing}, we may assume that the vertex set $V(G)$ is equal to $[n]$, and  $\Pp$ is a convex contraction sequence for the usual order on $[n]$. 
In particular, pairs of vertices can be identified with points in a plane,
and intuitively, every pair of sets $A,B\subset V(G)$ corresponds to a rectangle $A\times B\subset [n]\times [n]$. This correspondence will be important in the proof of Lemma~\ref{lem:proximity-ds}, whose 
 key technical component is the data structure for {\em{orthogonal range queries}} due to Chan~\cite{Chan11}, for manipulating rectangles in a plane.
(We remark that the applicability of this data structure in the context of twin-width has already been observed in~\cite{PilipczukSZ21}.) Let us recall the setting.

A {\em{rectangle}} is a set of pairs of integers of the form $\{(x,y)\colon a\leq x\leq a', b\leq y\leq b'\}$ for some integers $a,a',b,b'$. In all algorithmic statements that follow, every rectangle is represented by such a quadruple $(a,a',b,b')$. In the problem of orthogonal range queries, we are given a list of pairwise disjoint rectangles $\Rr=\{R_1,\ldots,R_m\}$, all contained in $[n]\times [n]$, and the task is to set up a data structure that can efficiently answer the following queries: given $(x,y)\in [n]\times [n]$, output the index of the rectangle in $\Rr$ that contains $(x,y)$, or output $\bot$ if there is no such rectangle. Chan proposed the following data structure for this problem.

\begin{theorem}[\cite{Chan11}]\label{thm:chan}
 Assuming $|\Rr|=\Oh(n)$, there is a data structure for the orthogonal range queries that takes $\Oh(n)$ space, can be initialized in time $\Oh(n)$, and can answer every query in time $\Oh(\log \log n)$.
\end{theorem}

We remark that there is also a simple data structure for orthogonal range queries that for any fixed $\varepsilon>0$, achieves query time $\Oh(1/\varepsilon)$ at the expense of space complexity and initialization time $\Oh(n^{1+\varepsilon})$. See the appendix of~\cite{PilipczukSZ21} for details. As we mentioned in Section~\ref{sec:intro}, replacing the usage of the data structure of Chan with this simple data structure results in an analogous tradeoff in Theorem~\ref{thm:intro-main-qa}.

We reduce the statement of Lemma~\ref{lem:proximity-ds} to the result of Chan using the following lemma.

\begin{lemma}\label{lem:break-rectangles}
 One can in time $\Oh_{d,r}(n)$ compute a list $\Qq$ of pairs of the form $(R,t)$, where $R\subseteq [n]\times [n]$ is a rectangle and $t\in [n]$, such that the following holds:
 \begin{itemize} 
  \item the rectangles in pairs from $\Qq$ form a partition of $[n]\times [n]$, and
  \item for each $(u,v)\in [n]\times [n]$, if $(R,t)\in \Qq$ is the unique pair satisfying $(u,v)\in R$, then we have $\prx_r(u,v)=t$.
 \end{itemize}
\end{lemma}
\begin{proof}
 We proceed through the contraction sequence, by considering times $t=2,3,\ldots,n$, maintaining the current impurity graph $\Imp_t$; this can be easily updated in time $\Oh_{d,r}(1)$ per each given time $1<t\le n$. Also, we maintain a list $\Qq$ with the following invariant: at the beginning of processing time $t$, $\Qq$ consists of pairs of the form $(R,s)$ for $s<t$ such that rectangles in $\Qq$ cover all pairs $(u,v)\in [n]\times [n]$ with $\prx_r(u,v)<t$. Note that at the beginning we can initialize $\Qq$ as $\{(\{(u,u)\},1)\colon u\in [n]\}$ and thus the invariant is maintained. Also, once we finish processing time $n$, the invariant tells us that the obtained list $\Qq$ can be output by the algorithm.
 
 We now implement a step of the process, say at a time $t>1$. Our goal is to find a set $\Rr_t$ of $\Oh_{d,r}(1)$ rectangles that cover all pairs $(u,v)\in [n]\times [n]$ with $\prx_r(u,v)=t$ and no other pair; then we can append $\{(R,t)\colon R\in \Rr_t\}$ to $\Qq$.
 
 Suppose $\Pp_t$ is obtained from $\Pp_{t-1}$ by contracting parts $A,A'\in \Pp_{t-1}$ into $B_t=A\cup A'\in \Pp_t$. Consider any $(u,v)\in [n]\times [n]$ such that $\prx_r(u,v)=t$. Observe that since $t$ is the first time when $\dist_t(\iwarp{u}{t},\iwarp{v}{t})\leq r$, it must be the case that
 $$\dist_t(\iwarp{u}{t},B_t)\leq r\qquad \textrm{and}\qquad \dist_t(\iwarp{v}{t},B_t)\leq r.$$
 Furthermore, for every $u'\in \iwarp{u}{t-1}$ and $v'\in \iwarp{v}{t-1}$ we also have $\prx_r(u',v')=t$. From these two observations we infer that $\Rr_t$ can be constructed as follows:
 \begin{itemize}
  \item Let $\Ff=\{U\in \Pp_t~|~\dist_t(B_t,U)\leq r\}\setminus \{B_t\}\cup \{A,A'\}$. Note that $|\Ff|\leq \Oh_{d,r}(1)$ and $\Ff$ can be constructed in time $\Oh_{d,r}(1)$.
  \item For every pair $\{U,U'\}\in \binom{\Ff}{2}$, verify whether $\dist_{t-1}(U,U')>r$ and $\dist_{t}(U,U')\leq r$ (in the second condition we replace $U$, respectively $U'$, with $B_t$ in case it belongs to $\{A,A'\}$). If this is the case, then add $U\times U'$ to $\Rr_t$. Note that since $\Pp$ is convex, $U\times U'$ is indeed a rectangle.
 \end{itemize}
 Observe that the distance conditions used in the second point above can be checked in time $\Oh_{d,r}(1)$ by running BFS trimmed at depth $r$ from $U$ in graphs $\Imp_{t-1}$ and $\Imp_t$. So every step can be executed in time $\Oh_{d,r}(1)$, giving running time $\Oh_{d,r}(n)$ in total.
\end{proof}

Lemma~\ref{lem:proximity-ds} follows from Lemma~\ref{lem:break-rectangles} as follows. Let $\Qq$ be the list provided by Lemma~\ref{lem:break-rectangles}; note that $|\Qq|\leq \Oh_{d,r}(n)$, because this is an upper bound on the running time of the algorithm computing $\Qq$. Let $\Rr$ be the list of rectangles appearing in the pairs from $\Qq$. Set up a data structure of Theorem~\ref{thm:chan} for $\Rr$ and, additionally, for each $R\in \Rr$ remember the unique $t\in [n]$ such that $(R,t)\in \Qq$. Then upon query $(u,v)\in [n]\times [n]$, it suffices to use the data structure of Theorem~\ref{thm:chan} to find the unique $R\in \Rr$ containing $(u,v)$ and return the associated integer $t$.


\subsection{The tree of $r$-close $\tup x$-tuples}\label{sec:warping}

In this section we are going to construct an auxiliary data structure for handling local types.
Fix a number $k\in \N$; this is the quantifier rank of the types we would like to tackle. Denote $r:=2^k$.
Also fix a finite set $\tup x$ of variables, an $n$-vertex graph $G$,
together with a contraction sequence $\Pp_1,\ldots,\Pp_n$.

For $s\in [n]$ and a tuple $\tup u\in \Pp_s^{\tup x}$, we call $\tup u$ {\em{$r$-close}} at the time $s$ if one cannot partition $\tup u$ into two nonempty tuples $\tup u',\tup u''$ such that $\dist_s(\tup u',\tup u'')>r$. Equivalently, if one considers an auxiliary graph on vertex set $\tup u$ where two parts are connected iff they are at distance at most $r$ in $\Imp_s$, then $\tup u$ is $r$-close iff this auxiliary graph is connected. Note that if $\tup u\in \Pp_s^{\tup x}$ is $r$-close at the time $s$, then for every $t$ with $s\leq t\leq n$, the tuple $\warp{\tup u}{s}{t}$ is also $r$-close at the time $t$.
%


\medskip
For $s\in [n]$ with $s>1$, by $B_s$ denote the part of $\Pp_s$ that is the union of two parts in $\Pp_{s-1}$.
Let $T_{r,\tup x}$ 
be the set consisting of all pairs of the form $(\tup u,s)$ such that $s\in [n]$, $\tup u\in \Pp_s^{\tup x}$ is $r$-close at the time~$s$, and at least one of the following conditions is satisfied:
\begin{itemize} 
 \item $s=1$; or
 \item $s>1$ and $\dist_s(B_s,\tup u)\leq r$; or
 \item $s<n$ and $\dist_{s+1}(B_{s+1},\warp{\tup u}{s}{s+1})\leq r$.
\end{itemize}

Note that as $\tup u$ is assumed to be $r$-close, if the second condition holds then $\tup u\subseteq \Vic_s^{r|\tup x|}(B_s)$, and if the third condition holds then $\warp{\tup u}{s}{s+1}\subseteq \Vic_{s+1}^{r|\tup x|}(B_{s+1})$. Since the trigraphs $ \Vic_s^{r|\tup x|}(B_s)$ and $\Vic_{s+1}^{r|\tup x|}(B_{s+1})$ are of size $\Oh_{d,k,\tup x}(1)$, it follows that $T_{r,\tup x}$ contains $\Oh_{d,k,\tup x}(n)$ elements in total: $n$ elements for $s=1$ and $\Oh_{d,k,\tup x}(1)$ elements for each $1<s\leq n$. 

We  consider an ancestor relation $\preceq$ on $T_{r,\tup x}$ defined as follows:
$$(\tup v,t)\preceq (\tup u,s)\qquad\textrm{if and only if}\qquad s\leq t\textrm{ and }\warp{\tup u}{s}{t}=\tup v.$$
It is easy to see $T_{r,\tup x}$ together with $\preceq$ defines a rooted tree whose tree order is $\preceq$.
The root is $(\tup r,n)$, where $\tup r$ is the unique tuple of $\Pp_n^{\tup x}$, the one that maps all variables of $\tup x$ to the unique part of $\Pp_n$. From now on we identify the set $T_{r,\tup x}$ with the tree it induces. Therefore, we call the elements of $T_{r,\tup x}$ {\em{nodes}} and the child-parent pairs in $T_{r,\tup x}$ the {\em{edges}} of $T_{r,\tup x}$.

\begin{definition}\label{def:close-tree}
    We call $T_{r,\tup x}$ the \emph{tree of $r$-close $\tup x$-tuples}
    associated with $G$ and the contraction sequence $\Pp_1,\ldots,\Pp_n$.
\end{definition}

Recall that $r=2^k$, and $k\in\N$ is fixed.
For every node $(\tup u,s)\in T_{r,\tup x}$, let $\Types^k_{\tup u,s}$ be the set of all possible $k$-local types of tuples $\tup w\in V(G)^{\tup x}$ satisfying $\tup u=\iwarp{\tup w}{s}$. By Lemma~\ref{lem:ltp_number_of_types}, there is a constant $M=\Oh_{d,k,\tup x}(1)$ such that $|\Types^k_{\tup u,s}|\leq M$ for every node $(\tup u,s)$, and $\Types^k_{\tup u,s}$ can be computed in time $\Oh_{d,k,\tup x}(1)$ given access to $\Tri_s$ and $\tup u$.

Consider nodes $(\tup u,s),(\tup v,t)\in T_{r,\tup x}$ such that $(\tup v,t)$ is the parent of $(\tup u,s)$. Let $e=((\tup u,s),(\tup v,t))$ be the corresponding edge of $T_{r,\tup x}$. By Lemma~\ref{lem:tp-warp}, there exists a function $f_e\colon \Types^k_{\tup u,s}\to \Types^k_{\tup v,t}$ such that for every tuple $\tup w\in V(G)^{\tup x}$ with $\tup u=\iwarp{\tup w}{s}$, we have
\begin{equation}\label{eq:bobr}
\ltp_t^k(\tup w)=f_e(\ltp_s^k(\tup w)). 
\end{equation}
We now verify that all the objects introduced above can be computed efficiently. 

\newcommand{\Close}{\mathsf{Close}}

\begin{lemma}\label{lem:construct-tree}
 One can in time $\Oh_{d,k,\tup x}(n)$ compute the nodes and the edges of $T_{r,\tup x}$ (where $r=2^k$) as well as, for every edge $e$ of $T_{r,\tup x}$, the function $f_e$.
\end{lemma}
\begin{proof}
 Recall that we denote $r=2^k$. For every time $t\in [n]$, let $\Close_t$ denote the set of all $r$-close tuples a time $t$. Observe that we can in time $\Oh_{d,k,\tup x}(n)$ scan through the contraction sequence $\Pp_1,\ldots,\Pp_n$ while maintaining, at every time $t$, both the impurity graph $\Imp_t$ and the set $\Close_t$. Indeed, computing $\Imp_{t+1}$ and $\Close_{t+1}$ from $\Imp_{t}$ and $\Close_t$ only requires inspecting the ball of radius $r(|\tup x|+1)$ around the part resulting from the contraction, which can be done in time $\Oh_{d,k,\tup x}(1)$ per time step. The nodes of $T_{r,\tup x}$ corresponding to time $t$ can be easily computed in time $\Oh_{d,k,\tup x}(1)$ from $\Close_t$ and the knowledge of the contraction at time $t$ right from the definition. So we already see that the nodes of $T_{r,\tup x}$ can be computed in total time $\Oh_{d,k,\tup x}(n)$.
 
 It remains to argue how to augment the scan explained above so that along the way, we also compute the edges of $T_{r,\tup x}$ and the corresponding functions $f_e$. It is easy to see that there are two types of edges in $T_{r,\tup x}$:
 \begin{itemize}
  \item Edges of the form $((\tup u,s),(\tup v,s+1))$, where $\tup v=\warp{\tup u}{s}{s+1}$, such that $\dist(B_{s+1},\tup v)\leq r$.
  \item Edges of the form $((\tup u,s),(\tup u,t))$ for $s<t$ such that $\dist(B_{s'},\tup u)>r$ for for all $s<s'\leq t$. (In particular, $\tup u\in \Pp_{s'}^{\tup x}$ for all such $s'$.)
 \end{itemize}
Edges of the second type can be constructed easily during the scan when moving from time $s$ to time $s+1$. The associated functions $f_e$ can be computed using Lemma~\ref{lem:tp-warp}. Edges of the first type can be constructed by remembering during the scan, for each $\tup u\in \Close_t$, the unique already computed node $(\tup u,s)\in T_{r,\tup x}$ with the largest $s$. Then, when we decide that $(\tup u,s)\in T_{r,\tup x}$, we simply add the edge $e=((\tup u,s),(\tup u,t))$. By Lemma~\ref{lem:relevant}, in this case we can set $f_e$ to be the identity function.
\end{proof}

We can finally state and prove the main result of this section.

\begin{lemma}\label{lem:warp-ds}
 One can in time $\Oh_{d,k,\tup x}(n)$ construct a data structure that can answer the following queries in time $\Oh_{d,k,\tup x}(1)$: given $\tup w\in V(G)^{\tup x}$, two nodes $(\tup v,t)\preceq (\tup u,s)$ of $T_{r,\tup x}$ such that $\tup u=\iwarp{\tup w}{s}$ and $\tup v=\iwarp{\tup w}{t}$, and the type $\ltp_s^k(\tup w)$, output the type $\ltp_t^k(\tup w)$. 
\end{lemma}
\begin{proof}
 Using Lemma~\ref{lem:construct-tree} construct the tree $T_{r,\tup x}$ and functions $f_e$ for the edges of $T_{r,\tup x}$. By Lemma~\ref{lem:ltp_number_of_types}, there is a constant $M=\Oh_{d,k,\tup x}(1)$ such that $|\Types^k_{\tup u,s}|\leq M$ for every node $(\tup u,s)$.
  Let $I\coloneqq [M]$ be an indexing set of size $M$. Since for every  node $(\tup u,s)$ we have $|\Types^k_{\tup u,s}|\leq M$, we can set an arbitrary injection $\iota_{\tup u,s}\colon \Types^k_{\tup u,s}\to I$. For an edge $e=((\tup u,s),(\tup v,t))$, we set
 $$g_e\coloneqq \iota_{\tup u,s}^{-1}\, ;\, f_e\, ;\, \iota_{\tup v,t}.$$
 Thus, $g_e$ is a function from $I$ to $I$ that, intuitively, is just $f_e$ reindexed using the index set $I$. Clearly, functions $\iota_{\tup u,s}$ and $g_e$ defined above can be computed in total time $\Oh_{d,k,\tup x}(n)$.
 
 We will use the following result proved in~\cite{PilipczukSSTV21}.
 
 \begin{theorem}[Theorem~5.1 of~\cite{PilipczukSSTV21}]\label{thm:semigroup}
  Let $S$ be a semigroup and $T$ be a rooted tree with edges labelled with elements of $S$. Then one can in time $|S|^{\Oh(1)}\cdot |T|$ construct a data structure that can answer the following queries in time $|S|^{\Oh(1)}$: given nodes $u,v\in T$ such that $v$ is an ancestor of $u$, output the (top-down) product of elements of $S$ associated with the edges on the path from $v$ to $u$. The data structure uses $|S|^{\Oh(1)}\cdot |T|$ space.
 \end{theorem}

 Let $S$ be the semigroup of functions from $I$ to $I$ with the product defined as $f\cdot g=g;f$. Thus, the functions $g_e$ form a labelling of edges of $T_{r,\tup x}$ with elements of $S$. Apply Theorem~\ref{thm:semigroup} to this $S$-labelled tree, and let ${\mathbb{S}}$ be the obtained data structure. Now, to answer a query about nodes $(\tup u,s)$, $(\tup v,t)$, and type $\alpha=\ltp_s^k(\tup w)$ as in the lemma statement, it suffices to apply the following procedure:
 \begin{itemize} 
  \item Compute $\tilde{\alpha}\coloneqq \iota_{\tup u,s}(\alpha)$.
  \item Query ${\mathbb{S}}$ to compute the compositions of functions $g_e$ along the path from $(\tup u,s)$ to $(\tup v,t)$ in $T_{r,\tup x}$. Call the resulting function $h$.
  \item Compute $\tilde{\beta}\coloneqq h(\tilde{\alpha})$.
  \item Output $\beta\coloneqq \iota_{\tup v,t}^{-1}(\tilde \beta)$.
 \end{itemize}
 The correctness of the procedure follows from a repeated use of~\eqref{eq:bobr}, and it is clear that the running time is $\Oh_{d,k,\tup x}(1)$.
\end{proof}

\subsection{Data structure}
%

With all the tools prepared, we can prove Theorem~\ref{thm:intro-main-qa}.

%


 Let $k$ be the quantifier rank of $\varphi$. We set up two auxiliary data structures:
\begin{itemize} 
 \item The data structure of Lemma~\ref{lem:proximity-ds} for radius parameter $r=2^k$. Call this data structure $\Prx$.
 \item For every $\tup z\subseteq \tup x$, the data structure of Lemma~\ref{lem:warp-ds} for parameter $k$ and the set of variables~$\tup z$. Call this data structure $\Wrp_{\tup z}$.
\end{itemize}
Moreover, using Lemma~\ref{lem:compute-affected}, we compute for each time $s\in [n-1]$ the trigraph $\Rel_s^{p}$, where $p\coloneqq r(|\tup x|+1)$. These objects constitute our data structure, so by Lemmas~\ref{lem:proximity-ds},~\ref{lem:warp-ds}, and~\ref{lem:compute-affected}, the construction time is $\Oh_{d,\varphi}(n)$ as promised. It remains to show how to implement queries.

Suppose we are given a tuple $\tup w\in V(G)^{\tup x}$ and we would like to decide whether $G\models \varphi(\tup w)$. By Lemma~\ref{lem:ltp_type_as_local_type}, to answer this it suffices to compute $\ltp^k_n(\tup w)$. In the following, for $\tup z\subseteq \tup x$, by $\tup w_{\tup z}$ we denote the restriction of $\tup w$ to the variables of $\tup z$.

For each time $s\in [n]$, let $H_s$ be the graph on vertex set $\tup x$ such that $y,y'\in \tup x$ are adjacent in $H_s$ if and only if $\dist_s(\iwarp{\tup w}{s}(y),\iwarp{\tup w}{s}(y'))\leq r$. The following are immediate:
\begin{itemize} 
 \item For all $1\leq s\leq t\leq n$, $H_t$ is a supergraph of $H_s$. That is, if $y,y'\in \tup x$ are adjacent in $H_s$, then they are also adjacent in $H_t$.
 \item If $\tup z\subseteq \tup x$ is such that $H_s[\tup z]$ is connected for some $s\in [n]$, then $\iwarp{\tup w_{\tup z}}{s}$ is $r$-close at the time~$s$.
\end{itemize}
Using the data structure $\Prx$, we may compute $\prx_{r}(y,y')$ for all $\{y,y'\}\in \binom{\tup x}{2}$ in total time $\Oh_{d,\varphi}(\log \log n)$. Let $S\subseteq [n]$ be the set of all those numbers, and include $1$ and $n$ in $S$ in addition. Thus $|S|\leq 2+\binom{|\tup x|}{2}\leq \Oh_{\varphi}(1)$. We imagine $S$ as ordered by the standard order $\leq$, hence we may talk about consecutive elements of $S$.

Note that the knowledge of the numbers $\prx_{r}(y,y')$ for $\{y,y'\}\in \binom{\tup x}{2}$ allows us to compute the graphs $H_{s}$ for all $s\in S$. Further, observe that if $t\in [S]$ is such that $s\leq t<s'$ for some $s,s'\in S$ that are consecutive in $S$, then $H_t=H_s$.

Let $L$ be the set of all pairs of the form $(\tup z,s)$ where $s\in S$, $\tup z$ is a connected component of $H_s$, and either $s=1$ or $\tup z$ is not connected in $H_{s-1}$. Clearly, $L$ has size $\Oh_{d,\varphi}(1)$ and can be computed in time $\Oh_{d,\varphi}(1)$. Observe the following.

\begin{claim}\label{cl:pinned}
 For each $(\tup z,s)\in L$, we have $(\iwarp{\tup w_{\tup z}}{s},s)\in T_{r,\tup z}$. If moreover $s>1$, then for every $\tup y\subseteq \tup z$ that is a connected component of $H_{s-1}$, we have $(\iwarp{\tup w_{\tup y}}{s-1},s-1)\in T_{r,\tup y}$.
\end{claim}
\begin{proof}
 If $s=1$, then $\tup z$ being a connected component of $H_1$ means that $\tup z$ is a constant tuple. Hence $\iwarp{\tup w_{\tup z}}{1}$ is $r$-close at the time $1$, implying that $(\iwarp{\tup w_{\tup z}}{1},1)\in T_{r,\tup z}$.
 
 Assume then that $s>1$.
 As $\tup z$ is not connected in $H_{s-1}$ and is connected in $H_s$, it follows that for every connected component $\tup y$ of $H_{s-1}$ that is contained in $\tup z$, we have $\dist_s(\iwarp{\tup w_{\tup y}}{s},B_s)\leq r$, and in particular $\dist_s(\iwarp{\tup w_{\tup z}}{s},B_s)\leq r$. The latter statement implies that $(\iwarp{\tup w_{\tup z}}{s},s)\in T_{r,\tup z}$ due to fulfilling the second condition in the definition of $T_{r,\tup z}$. Further, since $\tup y$ is a connected component of $H_{s-1}$, $\iwarp{\tup w_{\tup y}}{s-1}$ is $r$-close at the time $s-1$. So $(\iwarp{\tup w_{\tup y}}{s-1},s-1)\in T_{r,\tup y}$ due to fulfilling the third condition in the definition of $T_{r,\tup y}$.
\cqed\end{proof}

We now compute the types $\ltp_s^k(\tup w_{\tup z})$ for all $(\tup z,s)\in L.$ We do this in any order on $L$ with non-decreasing $s$, hence when processing $(\tup z,s)$ we may assume that the corresponding types have already been computed for all $(\tup y,t)\in L$ with $t<s$.

Assume first that $s=1$. Then $(\tup z,1)\in L$ means that $\tup z$ is a connected component of $H_1$, which in turn means that $\tup w_{\tup z}$ is a constant tuple.
In this case $\ltp_1^k(\tup w_{\tup z})$ can be computed trivially.

Assume then that $s>1$. Since $(\tup z,s)\in L$, we have that $\tup z$ is a connected component of $H_s$, but in $H_{s-1}$, $\tup z$ breaks into two or more smaller connected components. 

Consider any such component $\tup y$; that is, $\tup y$ is a connected component of $H_{s-1}$ that is contained in $\tup z$. By Claim~\ref{cl:pinned}, we have $(\iwarp{\tup w_{\tup y}}{s-1},s-1)\in T_{r,\tup y}$. Let then $t\leq s-1$ be the smallest time such that $\tup y$ is a connected component of $H_{t}$; clearly we have $t\in S$ and $(\tup y,t)\in L$. By Claim~\ref{cl:pinned} again, $(\iwarp{\tup w_{\tup y}}{t},t)\in T_{r,\tup y}$. Since the type $\ltp_{t}^k(\tup w_{\tup y})$ has been already computed before, we may use one query to $\Wrp_{\tup y}$ to compute the type $\ltp_{s-1}^k(\tup w_{\tup y})$.

Having performed the procedure described above for every connected component $\tup y$ of $H_{s-1}$ that is contained in $\tup z$, we may repeatedly use Lemma~\ref{lem:tp-merge} to compute the type $\ltp_{s-1}^k(\tup w_{\tup z})$. Note that for different components $\tup y,\tup y'$ as above, we have $\dist_{s-1}(\iwarp{\tup w_{\tup y}}{s-1},\iwarp{\tup w_{\tup y'}}{s-1})>r$ due to $\tup y$ and $\tup y'$ being non-adjacent in $H_{s-1}$. Furthermore, all trigraphs required in the applications of Lemma~\ref{lem:tp-merge} can be easily deduced from the trigraph $\Rel_{s-1}^p$ and the description of the contraction performed at the time $s-1$; these are stored in our data structure.

Finally, it remains to apply Lemma~\ref{lem:tp-warp} to compute the type $\ltp_s^k(\tup w_{\tup z})$ from $\ltp_{s-1}^k(\tup w_{\tup z})$. Again, the trigraphs needed in this application can be easily deduced from $\Rel_{s-1}^p$ and the contraction performed at the time $s-1$. 
This finishes the computation of types $\ltp_s^k(\tup w_{\tup z})$ for all $(\tup z,s)\in L$; note that the running time is $\Oh_{d,\varphi}(1)$. 

Finally, let $t$ be the smallest time such that $\tup x$ is connected in $H_{t}$. Such $t$ exists since $\tup x$ is connected in $H_n$. Clearly, $t\in S$. By definition we have $(\tup x,t)\in L$, so the type $\ltp_{t}^k(\tup w)$ has been computed. By Claim~\ref{cl:pinned}, $(\iwarp{\tup w}{t},t)\in T_{r,\tup x}$. So we can now use the data structure $\Wrp_{\tup x}$ one last time to compute $\ltp_n^k(\tup w)$. This finishes the proof of Theorem~\ref{thm:intro-main-qa}.

\section{Query enumeration}\label{sec:enumeration}

In this section we prove Theorem~\ref{thm:intro-main-qe}, which we recall below for convenience.
\newenvironment{red}
  {\color{red}%
   \ignorespaces}
 {}

\Enumeration*

First, we need to define our notion of enumerators and prepare some tools for them.

\paragraph{Enumerators.}
Let $x_1,\ldots,x_n$ be a sequence of elements.
An \emph{enumerator} of the sequence $x_1,\ldots,x_n$
is a data structure that implements a single method,
such that at the $i$th invocation of the method,
it outputs the element $x_j$ of the sequence,
where $j={i\mod n}$,
and reports an `end of sequence' message if $j=0$.
We say that the enumerator has \emph{delay} $t$ if 
each invocation takes at most $t$ computation steps,
including the steps needed to output the element $x_j$ 
(assuming each element has a fixed representation).
An enumerator for a set $X$ is an enumerator for 
any sequence $x_1,\ldots,x_n$ with $\set{x_1,\ldots,x_n}=X$
and $n=|X|$.
Enumerators 
 for Cartesian products and disjoint unions of sets
 can be obtained in an obvious way:

\begin{lemma}\label{lem:enum-prod}
    Suppose we are given an enumerator for a set $X$ 
    with delay $t$ and an enumerator for a set $Y$ with delay $t'$, where $t,t'\ge 1$.
    Then we can construct in time $\Oof(1)$ an  enumerator 
    with delay $t+t'+\Oof(1)$
    for the set $X\times Y$
    and
    -- if $X$ and $Y$ are disjoint --
    for the set $X\uplus Y$. 
\end{lemma}
We will also construct enumerators 
for disjoint unions of families of sets,  as follows.

\begin{lemma}\label{lem:enum-sum}
    Suppose $X_1,\ldots,X_n$ are pairwise disjoint, nonempty sets,
    such that $X_i$ has an enumerator $\cal E_i$ with delay $t$.
    Suppose furthermore we have an enumerator 
    for the sequence $\cal E_1,\ldots,\cal E_n$
    with delay $t'$.
    Then one can construct, in time $\Oof(1)$
    an enumerator for the set $\bigcup_{1\le i\le n} X_i$
    with delay $t+t'+\Oof(1)$.
\end{lemma}

Finally, we will use the following lemma, proved in \cite[Lemma 7.15]{PilipczukSSTV21}.
\begin{lemma}
    Fix a finite set $Q$ of size $q$ and a set of functions $\cal F\subset Q^Q$. There is a constant $c$ computable from $q$ and an algorithm that, given a rooted tree $T$, in which each edge $vw$ ($v$ child of $w$) is labelled by a function $f_{vw}\from Q\to Q$, computes in time $c\cdot |T|$ a collection $(\cal E_w)_{w\in V(T)}$ of enumerators, where each 
    $\cal E_w$ is an enumerator with delay $c$ that enumerates all descendants $v$ of $w$ such that the composition of the functions labeling the edges of the path from $v$ to $w$, belongs to $\Ff$.
\end{lemma}

This immediately yields the following.
\begin{corollary}\label{cor:composer-enum}
    Fix a number $q$. There is a constant $c$ computable from $q$ and an algorithm that, given a rooted tree $T$, 
    in which each node $v$ is labeled by a set $X_v$ with $|X_v|\le q$ and a set $Y_v\subset X_v$, and each     
    edge $vw$ ($v$ child of $w$) is labelled by a function $f_{wv}\from X_v\to X_w$, computes in time $c\cdot |T|$ a collection $(\cal E_w^\tau)_{w\in V(T),\tau\in X_w}$ of enumerators, where each
    $\cal E_w^\tau$ is an enumerator with delay $c$ that enumerates all descendants $v$ of $w$ such that 
    there is some $\sigma\in Y_v$ that is mapped to $\tau$ by the composition $f\from X_v\to X_w$ of the functions labeling the edges of the path from $v$ to $w$.
\end{corollary}


\paragraph{Proof of Theorem~\ref{thm:intro-main-qe}.}
We now proceed to the proof of Theorem~\ref{thm:intro-main-qe}.

Fix a number $k$, a set of variables $\tup x$,
 an $n$-vertex graph $G$, together with its contraction sequence $\Pp_1,\ldots,\Pp_n$.
 Denote $r:=2^k$. 

 For every $s\in [n]$ and $\tup x$-tuple $\tup u\in \Pp_s^{\tup x}$, and local type $\tau\in \Types_{\tup u,s}^k$, 
 denote 
\[S_{\tup u,s}^\tau\coloneqq \setof{\tup w\in V(G)^{\tup x}}{\iwarp{\tup w}{s}=\tup u
\text{ and } \ltp^k_s(\tup w)=\tau}.\]

Recall that the root of $T_{r,\tup x}$ is the pair $(\tup r,n)$,
where $\tup r$ is the constant $\tup x$-tuple with all components equal to the unique 
part of $\Pp_n$. Then $S_{\tup r,n}^\tau$ is the set of all $\tup x$-tuples $\tup w\in V(G)^{\tup x}$ 
with $\ltp^k_n(\tup w)=\tau$. 
From Lemma~\ref{lem:ltp_type_as_local_type} and Lemma~\ref{prop:types_mc}
we get:
\begin{lemma}\label{lem:enum wrap up}
    Fix a formula $\phi(\tup x)$ of quantifier-rank $k$.  
    Then there is a  set $\Gamma\subset \Types_{\tup r,n}^n$
    such that 
    $\phi(G):=\setof{\tup w\in V(G)^{\tup x}}{G\models \phi(\tup w)}$ is the disjoint union of 
the family of sets 
$\setof{S_{\tup r,n}^\tau}{\tau\in \Gamma}.$
\end{lemma}

Therefore, an enumerator for $\phi(G)$ can be obtained by concatenating 
enumerators for the sets $S_{\tup r,n}^\tau$, for $\tau\in \Gamma$.
Note that here we are concatenating only $\Oof_{k,d,\tup x}(1)$
enumerators, by Lemma~\ref{lem:ltp_number_of_types},
so, by applying  Lemma~\ref{lem:enum-prod} repeatedly,
the resulting enumerator can be obtained in time 
$\Oof_{k,d,\tup x}(1)$ and has delay $\Oof_{k,d,\tup x}(1)$.
So to prove Theorem~\ref{thm:intro-main-qe},
it suffices to prove that we can efficiently compute an enumerator 
for each of the sets $S_{\tup r,n}^\tau$.

Recall that $T_{r,\tup x}$ is the tree of $r$-close $\tup x$-tuples (see Def. \ref{def:close-tree}), and can be computed in time $\Oof_{d,k,\tup x}(n)$, by Lemma~\ref{lem:construct-tree}.
In the following proposition,
we will show how to compute enumerators for all of the sets $S_{\tup u,s}^\tau$,
for $(\tup u,s)\in T_{r,\tup x}$.
All the enumerators jointly will be computed in time $\Oof_{d,k,\tup x}(n)$.

\begin{proposition}\label{prop:enum}
    Fix a nonempty set $\tup x$ of variables and $k\in\N$.
    Assume $G$ is a graph on $n$ vertices provided on input through a contraction sequence 
    $\Pp_1,\ldots, \Pp_n$ of width $d$. 
    Then one can in time $\Oof_{d,k,\tup x}(n)$ construct a data structure that
    associates, to every node $(\tup u,s)$ of $T_{r,\tup x}$
    and every local type $\tau\in \Types_{\tup u, s}^k$,    
    an enumerator for all tuples in $S_{\tup u,s}^\tau$
    with delay $\Oof_{d,k,\tup x}(1)$.
\end{proposition}

As noted above, Theorem~\ref{thm:intro-main-qe} follows from 
Proposition~\ref{prop:enum}, using Lemma~\ref{lem:enum wrap up}.
The rest of Section~\ref{sec:enumeration} is devoted to proving 
Proposition~\ref{prop:enum}.

We  prove Proposition~\ref{prop:enum} by induction on $|\tup x|$.
So suppose the statement holds for all strict subsets of $\tup x$.
Recall that we may construct the tree $T_{r,\tup x}$,  in time $\Oof_{d,k,\tup x}(n)$, using Lemma~\ref{lem:construct-tree}.



Let $v,u$ be two nodes of $T_{r,\tup x}$ 
 with $v=(\tup v,t)$ and $u=(\tup u,s)$ and
$u\preceq v$.
By Lemma~\ref{lem:warp-ds}, there is a function 
$f_{vu}\from \Types_{\tup v,t}^k\to \Types_{\tup u,s}^k$
such that for every $\tup w\in V(G)^{\tup x}$, 
with $\tup u=\iwarp{\tup w}{t}$ we have
$f_{vu}(\ltp^k_t(\tup w))=\ltp^k_s(\tup w).$

\medskip
For a tuple $\tup w\in V(G)^{\tup x}$,
let $s\in [n]$
be the first time such that $\iwarp{\tup w}{s}$ is $r$-close at time $s$, where $r=2^k$.
We then say that $\tup w$ \emph{registers} 
at $(\tup u,s)$, where $\tup u=\iwarp{\tup w}{s}$.
By Claim~\ref{cl:pinned}, in this case, the pair $(\tup u,s)$ is a node of $T_{r,\tup x}$.

For each node $(\tup u,s)$ of $T_{r,\tup x}$
and type $\tau\in \Types_{\tup u,s}^k$,
denote:
\[R_{\tup u,s}^\tau=\setof{\tup w\in V(G)^{\tup x}}{\tup w \text{ registers at $(\tup u,s)$
and $\ltp^k_s(\tup w)=\tau$}}.\]

Fix a node $(\tup u,s)\in T_{r,\tup x}$ and a type 
    $\tau\in \Types_{\tup u,s}^k$.
    Clearly, every tuple  $\tup w\in S_{\tup u,s}^\tau$ 
    registers at exactly one descendant $v=(\tup v,t)$ of $u=(\tup u,s)$
    (possibly, $v=u$),
    and moreover, 
    ${f_{vu}(\ltp^k_t(\tup w))=\tau}$.
    This proves the following.

\begin{lemma}\label{l0}
    For every node $u\in T_{r,\tup x}$ and type 
    $\tau\in \Types_u^k$,
    the set $S_u^\tau$ 
    is the disjoint union of 
    all the sets $R_v^\sigma$, for 
    $v\in T_{r,\tup x}$ with $v\succcurlyeq u$ and $\sigma\in\Types_v^k$ such that
         $f_{vu}(\sigma)=\tau$.
\end{lemma}
We shall prove the following two lemmas.
\begin{lemma}\label{l1}
    For every given node $u=(\tup u,s)\in T_{r,\tup x}$ and type 
    $\tau\in \Types_{\tup u,s}^k$,
    an enumerator for the set $R_{\tup u,s}^\tau$ 
    with delay $\Oof_{d,k,\tup x}(1)$ can be constructed in time $\Oof_{d,k,\tup x}(1)$.
\end{lemma}

\begin{lemma}\label{l2}
    One can construct in time $\Oof_{d,k,\tup x}(n)$
    a collection of enumerators $\cal E_u^\tau$, 
one per each node $u=(\tup u,s)\in T_{r,\tup x}$
    and type $\tau\in \Types_{\tup u,s}^k$,
    where $\cal E_u^\tau$ has delay $\Oof_{d,k,\tup x}(1)$ and  enumerates all descendants 
    $v=(\tup v,t)$ of $u$ in $T_{r,\tup x}$
    such that there is some $\sigma\in \Types_{\tup v,t}^k$
    with 
     $f_{vu}(\sigma)=\tau$ and 
    $R_{\tup v,t}^\sigma\neq\emptyset$.
\end{lemma}

Observe that combining Lemma~\ref{l0}, Lemma~\ref{l1}, Lemma~\ref{l2} and Lemma~\ref{lem:enum-sum} yields the required collection of enumerators 
for each of the sets $S_{\tup u,s}^k$, 
thus proving Proposition~\ref{prop:enum} and Theorem~\ref{thm:intro-main-qe}.
Thus, we are left with proving Lemmas~\ref{l1} and~\ref{l2}, which we do in order.

\begin{proof}[Proof of Lemma~\ref{l1}]
 Let $u=(\tup u,s)\in T_{r,\tup x}$ be a node.
We consider two cases: either $s=1$ or $s>1$.

\subparagraph{Leaves.}
Consider the case that $s=1$.
As $\tup u\in \cal P_1^{\tup x}$ is $r$-close at the time $1$,
there is some $v\in V(G)$ such that 
all the components $\tup u$ are equal to the part $\set{v}$.
Then a tuple $\tup w\in V(G)^{\tup x}$ satisfies $\iwarp{\tup w}{1}=\tup u$ if and only if $\tup w$ is the constant $\tup x$-tuple $\vec v$ consisting of the vertex $v$.
Therefore,
\[R_{\tup u,s}^\tau=
\begin{cases}\set{\vec v}& \text{if }\ltp^k_1(\vec v)=\tau,\\
  \emptyset &\text{if }\ltp^k_1(\vec v)\neq \tau.
\end{cases}
\]
In either case, it is trivial to construct an enumerator for the (empty or singleton) set $R_{\tup u,s}^\tau$
and we can distinguish which case occurs by computing $\ltp^k_1(\vec v)$ in time $\Oof_{d,k,\tup x}(1)$.
The delay of the enumerator is bounded by $\Oof_{|\tup x|}(1)$,
as this is the size of the representation of the tuple $\vec v$.

\subparagraph{Inner nodes.}
Suppose now that $s>1$, that is, $(\tup u,s)$ is an inner node.
Denote 
\[\cal V\coloneqq \setof{\tup v\in \cal P_{s-1}^{\tup x}} {\iwarp{\tup v}{s-1\to s}=\tup u}.\]

\begin{claim}
    The set $\cal V$ has size $\Oof_{k,\tup x,d}(1)$ and can be computed in this time, given $(\tup u,s)\in T_{r,\tup x}$.
\end{claim}
\begin{proof}
    The tuples $\tup v$ in $\cal V$ are precisely those tuples that can be obtained from the tuple $\tup u$ by replacing each occurrence of $B_s$ in $\tup u$ (which might occur zero or more times) by one of the two parts in $\Pp_{s-1}$ that are contained in $B_s$.
    Since $B_s$ may occur in $\tup u$ at most $|\tup x|$ many times, 
    we have at most $2^{|\tup x|}$ many possibilities for $\tup v$,
    and all of them can be computed in the required time.
\cqed\end{proof}

Fix $\tup v\in \cal V$, and consider the graph $H_{\tup v}$ with vertices $\tup x$ 
where any two distinct $y,y'\in \tup x$ are adjacent whenever $\dist_s(\tup v(y),\tup v(y'))\le r$.
Let $C_{\tup v}$ denote the set of connected components of $H_{\tup v}$,
where each connected component is viewed as a set $\tup y\subset \tup x$ of vertices of $H_{\tup v}$. Then each of the sets $C_{\tup v}$, can be computed in time $\Oof_{k,\tup x,d}(1)$, given $(\tup u,s)\in T_{r,\tup x}$
and $\tup v\in \cal V$.

Call a tuple 
$\tup v\in \cal V$ \emph{disconnected} if $|C_{\tup v}|>1$, that is, $H_{\tup v}$ has more than one connected component.
Note that if $\tup v$ is disconnected and $\tup y\in C_{\tup v}$, then $|\tup y|<|\tup x|$.
Let $\cal V'\subset \cal V$ be the set of disconnected tuples.

\medskip
Fix an disconnected tuple $\tup v\in \cal V'$ and $\tup y\in C_{\tup v}$.
Denote by $\tup v_{\tup y}$ the restriction of $\tup v$ to $\tup y$.
Note that the pair $(\tup v_{\tup y},s-1)$ is a node of $T_{r,\tup y}$.
Indeed, by assumption,  $\dist_{s}(B_s,\tup u)\le r$ holds, so 
     $\dist_{s}(B_s,\iwarp{\tup v}{s\to s+1})\le r$,
    and in particular $\dist_{s}(B_s,\iwarp{\tup v_{\tup y}}{s\to s+1})\le r$.
Moreover, $\tup v_{\tup y}$ is $r$-close, since $\tup y$ is a connected component of 
$H_{\tup v}$.
As $|\tup y|<|\tup x|$, by inductive assumption, we have already computed enumerators for each of the sets 
$S_{\tup v_{\tup y},s-1}^\sigma$, for all adequate local types $\sigma$.

The following claim is obtained by repeatedly applying Lemma~\ref{lem:tp-merge}.
\begin{claim}\label{cl:step}
    For every $\tup w\in V(G)^{\tup x}$ with $\iwarp{\tup w}s=\tup u$,
the local type    $\ltp^k_s(\tup w)$ can be computed in time 
$\Oof_{k,d,\tup x}(1)$ from the following data:
    \begin{itemize}
        \item the tuple $\tup v:=\iwarp{\tup w}{s-1}\in \cal V$,
        \item the family of local types $\bar \tau:=(\ltp^k_{s-1}(\tup v_{\tup y}): \tup y\in C_{\tup  v})$.
    \end{itemize}
\end{claim}
    More precisely, there is a function $\Gamma$ 
    such that for each pair 
    $(\tup v,\bar \tau)$, where $\tup v\in \cal V$ and  $\bar \tau =(\tau_{\tup y}:\tup y\in C_\tup v)$
    is a family with  $\tau_{\tup y}\in\Types_{\tup v_{\tup y},s-1}^k$, we have that
     \[\Gamma(\tup v,\bar \tau)=\ltp^k_s(\tup w)\]
    holds for every $\tup w\in V(G)^{\tup x}$ such that  $\iwarp{\tup w}{s-1}=\tup v$
    and $\ltp^k_{s-1}(\tup v_{\tup y})=\tau_{\tup y}$ for all $\tup y\in C_\tup v$.

For $\tup v\in \cal V$ and 
$\bar \tau =(\tau_{\tup y}\colon \tup y\in C_\tup v)$
 a family with  $\tau_{\tup y}\in\Types_{\tup v_{\tup y},s-1}^k$,
 define the set 
 \[S_{\tup v,s-1}^{\bar \tau}\coloneqq
 \setof{\tup w\in V(G)^{\tup x}}
 {\ltp^k_{s-1}(\tup w_{\tup y})=\tau_{\tup y} \text{ for all }\tup y\in C_{\tup v}}
 .\]
 Recall that the tuples $\tup y\in C_{\tup v}$ form a partition of $\tup x$.
 Hence, we have the following.
 \begin{claim}\label{cl:prod}
    Fix an disconnected tuple $\tup v\in \cal V'$ and $\bar \tau =(\tau_{\tup y}:\tup y\in C_\tup v)$
 a family with  $\tau_{\tup y}\in\Types_{\tup v_{\tup y},s-1}^k$.
    One can compute in time $\Oof_{k,d,\tup x}(1)$ an enumerator for the set $S_{\tup v,s-1}^{\bar \tau}$,
    with delay $\Oof_{k,d,\tup x}(1)$.
 \end{claim}
 \begin{proof}
     The enumerator is the Cartesian product (see Lemma~\ref{lem:enum-prod}) of the enumerators for the sets $S_{\tup v_{\tup y},s-1}^{\tau_{\tup y}}$, which have been computed by inductive assumption. Each of those enumerators  has delay $\Oof_{k,d,\tup x}(1)$, and their total number is $\Oof_{k,d,\tup x}(1)$, by Lemma~\ref{lem:ltp_number_of_types}.
 \cqed\end{proof}

\begin{claim}
    The set $R_{\tup u,s}^\tau$
    is equal to the disjoint union of the family of sets 
    \[\setof{S_{\tup v,s-1}^{\bar \tau}}{(\tup v,\bar \tau)\in \Gamma^{-1}(\tau), \tup v\in \cal V'}.\]
\end{claim}
\begin{proof}
    We first show the right-to-left inclusion.
    
    Let 
     $\tup w\in S^{\bar \tau}_{\tup v,s-1}$.
    Then registers at $(\tup u,s)$ as $\iwarp{\tup w}{s-1}$ 
    is not $r$-close at time $s-1$ (because $\tup v\in \cal V'$), and $\iwarp{\tup w}{s}=\tup u$.
    Moreover, $\ltp_k^s(\tup w)=\tau$ since $(\tup v,\bar \tau)\in\Gamma^{-1}(\tau)$. This proves that $\tup w\in R_{\tup u,s}^\tau$.

    Conversely, let $\tup w\in R_{\tup u,s}$.
    Define $\tup v$ as $\iwarp{\tup w}{s-1}$.
    Since $\tup w$ registers at $(\tup u,s)$,
    it follows that $H_{\tup v}$ is disconnected.
    Hence, $\tup v\in \cal V'$. For each connected component $\tup y\subset \tup x$ of $H_{\tup v}$, let $\tau_\tup y=\ltp_k^s(\iwarp{\tup w_\tup y}{s-1})$, and let 
     $\bar \tau=(\tau_{\tup y}\colon \tup y\in C_\tup v)$.
    By construction, $\tup w\in S_{\tup v,s-1}^{\bar \tau}$.
    This proves the left-to-right inclusion.
Moreover, it is easy to see that the union is disjoint.
\cqed\end{proof}

Therefore, an enumerator for the set $R_{\tup u,s}^\tau$ above can be obtained  by concatenating the enumerators for the sets $S_{\tup v,s-1}^{\bar \tau}$, for $(\tup v,\bar \tau)\in \Gamma^{-1}(\tau)$ with $\tup v\in \cal V'$, and those can be computed in time $\Oof_{d,k,\tup x}(1)$ by Claim~\ref{cl:prod}.
Note that we are taking a disjoint union of at most  $\Oof_{d,k,\tup x}(1)$ sets, by Lemma~\ref{lem:ltp_number_of_types},
so the concatenation can be computed by repeatedly applying Lemma~\ref{lem:enum-prod}.
\end{proof}

\begin{proof}[Proof of Lemma~\ref{l2}]
Label each node $u=(\tup u,s)$ of $T_{r,\tup x}$ by the 
set $X_u:=\Types^k_{\tup u,s}$
and its subset $Y_u\subset X_u$
of all $\tau\in \Types^k_{\tup u,s}$
such that $R_{\tup u,s}^{\tau}$ is nonempty.
This can be computed in time $\Oof_{k,d,\tup x}(n)$,
by testing emptiness of the enumerators produced in Lemma~\ref{l1}.
Moreover, $|X_u|$ is bounded by $\Oof_{k,d,\tup x}$, by Lemma~\ref{lem:ltp_number_of_types}.

Label each child-parent edge $e=((\tup v,t),(\tup u,s))$ of $T_{r,\tup x}$ by the function $f_e\colon \Types^k_{\tup v,t}\to \Types^k_{\tup u,s}$ such that for every tuple $\tup w\in V(G)^{\tup x}$ with $\tup v=\iwarp{\tup w}{s}$, we have
\begin{equation}
\ltp_t^k(\tup w)=f_e(\ltp_s^k(\tup w)), 
\end{equation}
where we denote $\ltp^k_t(\cdot)\coloneqq \ltp^k_{\Pp_t}(\cdot)$ for brevity. Such a function exists by Lemma~\ref{lem:tp-warp}.

Note that if $v:=(\tup v,t)$ is a descendant of $u:=(\tup u,s)$,
then the composition of the functions $f_e$ along the edges $e$ of the path from $v$ to $w$
is a function $f_{vu}\from \Types^k_{\tup v,t}\to \Types^k_{\tup u,s}$
such that for every tuple $\tup w\in V(G)^{\tup x}$ with $\tup v=\iwarp{\tup w}{s}$, we have
\begin{equation}
\ltp_t^k(\tup w)=f(\ltp_s^k(\tup w)).
\end{equation}

We are now in the setting of Corollary~\ref{cor:composer-enum}.
Hence we can compute in time $\Oof_{k,d,\tup x}(n)$
a collection of enumerators $\cal E_u^\tau$,
where for each node $u=(\tup u,s)$ of $T_{r,\tup x}$
and type $\tau\in \Types^k_{\tup u,s}$, 
the enumerator $\cal E_w^\tau$ enumerates 
all descendants $v=(\tup v,t)$ of $u$ in $T_{r,\tup x}$
such that $f_{vu}^{-1}(\tau)\cap Y_v$ is nonempty.
Unravelling the definitions, this means 
that there is some $\sigma\in \Types^k_{\tup v,t}$
such that $f_{vu}(\sigma)=\tau$ and 
 $R_{\tup v,t}^\sigma$ is nonempty.
\end{proof}

\newcommand{\hsim}{\approx}
\newcommand{\colors}{\mathrm{Colors}}
\newcommand{\tups}{\mathrm{Tuples}}

\section{VC density}\label{sec:numtypes}

In this section we prove Theorem~\ref{thm:intro-main-nei}, which we restate below for convenience.

\VCdensity*

%
%

\subsection{Additional preliminaries}

We first need a few additional definitions and observations.

\paragraph*{Bipartite contraction sequences.}
Let $G$ be a graph on $n$ vertices and $A$ be a subset of its vertices.
We say that a contraction sequence $\seq \Pp n$ of $G$ is {\em{$A$-bipartite}} if at every time $t \in [n-1]$ for every part $B \in \Pp_t$ we have $B \cap A = \emptyset$ or $B \subseteq A$. In other words, every part of $\Pp_t$ consists only of vertices in $A$ or in $V(G) \setminus A$.
By \cite[Lemma 7]{tww1} we have that if the twin-width of $G$ is at most $d$, then for any $A \subseteq V(G)$ there exists an $A$-bipartite contraction sequence of width at most $2d$.

\paragraph*{\texorpdfstring{$r$}{r}-distance coloring.}
Let $r \in \N$ be an integer, $G$ be a graph on $n$ vertices, $A$ be a nonempty subset of its vertices, and $\seq \Pp n$ be an $A$-bipartite contraction sequence of $G$ of width at most $d$.
We define the {\em{$r$-distance coloring}} of $V(G) \setminus A$ as follows.
Initially we start with all the vertices in $V(G) \setminus A$ not colored.
Then, we analyze our contraction sequence step by step from $\Pp_1$ till $\Pp_n$.
If at some time $t\in [n-1]$, the partition $\Pp_{t + 1}$ is obtained by contracting two parts $C, C' \in \Pp_t$ such that $C, C' \subseteq A$ into $B=C\cup C' \in \Pp_{t+1}$, then we look at all the parts $Q \subseteq V(G) \setminus A$ such that the distance between $B$ and $Q$ in $\Imp_{t+1}$ is at most $r$.
For every such part $Q$ we color all the vertices in $Q$ that are not colored yet with a fresh color (different for every $Q$ as above).
Note that in this way, in one step we use $\Oh_{d,r}(1)$ fresh colors. 
There are exactly $|A| - 1$  times $t\in [n-1]$ where two parts contained in $A$ are merged, so in total we use $\Oh_{d,r}(|A|)$ different colors.
Finally, if there are some vertices in $V(G) \setminus A$ which we did not color in this procedure, we color them with one additional 
fresh color.

Note that an $r$-distance coloring is uniquely determined by the choice  of a specific $A$-bipartite contraction sequence, which will always be clear from the context.
By $\colors_r$ we denote the partition of $V(G) \setminus A$ into color classes in the $r$-distance coloring. We will often work with tuples of colors classes. If $\tup y$ is a finite set of variables and $\tup p\in \colors_r^{\tup y}$ is a $\tup y$-tuple of colors, then we define 
    $$\tups_{\tup p} \coloneqq \set{\,\tup b \in V(G)^{\tup y}~|~\tup b(y) \in \tup p(y)\textrm{ for all }y\in \tup y\,}.$$

\newcommand{\meet}{\mathsf{meet}}

\paragraph*{Meeting time.}
Let $G$ be a graph on $n$ vertices and $\seq \Pp n$ be a contraction sequence of~$G$.
For any non-empty $X \subseteq V(G)$ we define the {\em{meeting time}} of $X$, denoted $\meet(X)$, to be the earliest time $t\in [n]$ such that all vertices of $X$ are in the same part of $\Pp_t$. Since in $\Pp_n$ all vertices are in the same part, this is well-defined. 


\subsection{A technical lemma}

We now formulate a lemma that will be the main ingredient in the proof of Theorem \ref{thm:intro-main-nei}.

\begin{lemma}\label{lem:bounded_number_of_ltp}
    Let $\tup x, \tup y$ be finite sets of variables, $k \in \N$ be a fixed integer, $G$ be a graph on $n$ vertices, $A$ be a nonempty subset of its vertices, and $\seq \Pp n$ be an $A$-bipartite contraction sequence of $G$ of width at most $d$.
    Let $r \coloneqq 2^k \cdot ({|\tup x|} + {|\tup y|} - 1)$ and fix any tuple of color classes $\tup p \in \colors_r^{\tup y}$. Suppose $t$ is a time such that $t\geq \meet(\tup p(y))$ for all $y\in \tup y$.
    Define an equivalence relation $\sim$ on $\tups_{\tup p}$ as follows:
    \[
    \tup b \sim \tup b' 
    \qquad\textrm{if and only if}\qquad \ltp_t^k(\tup a\tup b) = \ltp_t^k(\tup a\tup b')\textrm{ for all }\tup a \in A^{\tup x}.
    \]
    Then the number of equivalence classes of $\sim$ is $\Oh_{d,k,{\tup x}, {\tup y}}(1)$.
\end{lemma}

Note that in the statement above, the equivalence relation $\sim$ depends on the choice of $\tup p$ and~$t$. We will use the notation $\sim$ only when $\tup p$ and $t$ are clear from the context. Further, observe that the assumption $t\geq \meet(\tup p(y))$ for all $y\in \tup y$ means that for every $y\in \tup y$ there is a unique part $\tup v(y)\in \Pp_t$ such that $\tup p(y)\subseteq \tup v(y)$. Therefore, the two local types considered in the statement both belong to the same type set $\Types^k_{\tup u\tup v,t}$, where $\tup u=\iwarp{\tup a}{t}$.

Before we proceed to the proof of Lemma~\ref{lem:bounded_number_of_ltp}, let us verify that it implies Theorem~\ref{thm:intro-main-nei}.

\begin{proof}[Proof of Theorem \ref{thm:intro-main-nei} assuming Lemma~\ref{lem:bounded_number_of_ltp}]
First, we will show
\begin{equation}\label{eq:wydra}
    |\,\{\,\{\,\tup a\in A^{\tup x}~|~G\models \varphi(\tup a,\tup b)\,\}~\colon~\tup b\in (V(G) \setminus A)^{\tup y}\,\}\,| = \Oh_{d, \varphi}\left(|A|^{|\tup y|}\right).
\end{equation}
Let us take any $A$-bipartite contraction sequence for $G$ of width at most $2d$.
Let $k$ be the quantifier rank of $\phi$ and set $r \coloneqq 2^k \cdot (|\tup x| + |\tup y| - 1)$.
Consider any $\tup p \in \colors_r^{\tup y}$ and let $\tup b, \tup b' \in \tups_{\tup p}$ be two tuples of vertices such that $\tup b \sim \tup b'$, where $\sim$ is the  equivalence relation defined in Lemma~\ref{lem:bounded_number_of_ltp} for $\tup p$ and $t=n$.
By the definition of $\sim$ and Lemma~\ref{lem:ltp_type_as_local_type},
\[
 \tp^k(\tup a\tup b) = \tp^k(\tup a\tup b')\qquad \textrm{for every }\tup a \in A^{\tup x}.\]
Therefore,
\[
\set{\tup a \in A^{\tup x}~|~G \models \phi(\tup a,\tup b)} = \set{\tup a \in A^{\tup x}~|~G \models \phi(\tup a,\tup b')}.
\]
By Lemma~\ref{lem:bounded_number_of_ltp}, $\sim$ has at most $\Oh_{d,k,\tup x,\tup y}(1)$ equivalence classes. We conclude that
$$|\,\{\,\{\,\tup a\in A^{\tup x}~|~G\models \varphi(\tup a,\tup b)\,\}~\colon~\tup b\in \tups_{\tup p}\,\}\,| = \Oh_{d, \varphi}\left(1\right)\qquad \textrm{for every }\tup p\in \colors_r^{\tup y}.$$
Now~\eqref{eq:wydra} follows from the facts that $|\colors_r^{\tup y}| = \Oh_{d, \phi}(|A|^{|\tup y|})$ and $\{\tups_{\tup p}\colon \tup p\in \colors_r^{\tup y}\}$ is a partition of $(V(G) \setminus A)^{\tup y}.$

We now proceed with the general case where we consider all tuples $\tup b \in V(G)^{\tup y}$, that is, the variables in $\tup y$ can be also mapped to vertices in $A$.
Consider any partition of $\tup y$ into $\tup z, \tup z'$.
By~\eqref{eq:wydra}, we already know that
\[
    |\,\{\,\{\,(\tup a\in A^{\tup x}, \tup c' \in A^{\tup z'})~|~G\models \varphi(\tup a,\tup c\tup c')\,\}~\colon~\tup c\in (V(G) \setminus A)^{\tup z}\,\}\,| = \Oh_{d, \varphi}\left(|A|^{|\tup z|}\right).
\]
In particular, for every fixed tuple $\tup c' \in A^{\tup z'}$ we have
\[
    |\,\{\,\{\,\tup a\in A^{\tup x}~|~G\models \varphi(\tup a,\tup c\tup c')\,\}~\colon~\tup c\in (V(G) \setminus A)^{\tup z}\,\}\,| = \Oh_{d, \varphi}\left(|A|^{|\tup z|}\right).
\]
There are $|A|^{|\tup z'|}$ possible tuples $\tup c'\in A^{\tup z'}$, hence we get
\[
    |\,\{\,\{\,\tup a\in A^{\tup x}~|~G\models \varphi(\tup a,\tup c\tup c')\,\}~\colon~\tup c\in (V(G) \setminus A)^{\tup z}, \tup c' \in A^{\tup z'}\,\}\,| = \Oh_{d, \varphi}\left(|A|^{|\tup z| + |\tup z'|}\right) = \Oh_{d, \varphi}\left(|A|^{|\tup y|}\right).
\]
As the number of possible partitions of $\tup y$ into $\tup z$ and $\tup z'$ is $2^{|\tup y|}=\Oh_\varphi(1)$, the lemma follows.
\end{proof}

\subsection{A single simplification step}
In preparation for the proof of Lemma~\ref{lem:bounded_number_of_ltp}, we show the following statement, which eventually will allow us to apply induction on the total number of involved variables. We use the convention that if $\tup a$ is an $\tup x$-tuple and $\tup z\subseteq \tup x$, then $\tup a_{\tup z}$ denotes the restriction of $\tup a$ to the variables~of~$\tup z$. 

\begin{lemma}
\label{lemma:technical_lemma}
Let $\tup x, \tup y$ be finite sets of variables, $m\in \tup x$ be a variable in $\tup x$, $k \in \N$ be a fixed integer, $G$ be a graph on $n$ vertices, $A'\subseteq A$ be two subsets of vertices of $G$ with $|A'|\geq 2$, and $\seq \Pp n$ be an $A$-bipartite contraction sequence of $G$ of width at most $d$.
Let $r \coloneqq 2^k \cdot ({|\tup x|} + {|\tup y|} - 1)$ and consider any tuple of color classes $\tup p \in \colors_r^{\tup y}$.  Let $t\coloneqq \meet(A')$
and $\tup z,\tup z'$ be the partition of $\tup y$ such that $\meet(\tup p(z))\leq t$ for all $z\in \tup z$ and $\meet(\tup p(z'))>t$ for all $z'\in \tup z'$.
Then for any tuples $\tup b \in \tups_{\tup p}$ and $\tup a\in A^{\tup x}$ such that $\tup a(m)\in A'$, the local type $\ltp_t^k(\tup a\tup b)$ depends only on $\ltp_{t}^k(\tup a_{\tup x\setminus \{m\}}\tup b)$ and $\ltp_{t}^k(\tup a\tup b_{\tup z})$.
\end{lemma}
\begin{proof}
As $|A'| \ge 2$, we have~$t>1$.
Hence, the vertices of $A'$ are in exactly two different parts of $\Pp_{t - 1}$ which get contracted into a single part $B \in \Pp_t$. In particular, $A'\subseteq B\subseteq A$.

Consider any $z' \in \tup z'$ and any $u\in \tup p(z')$. We observe that the distance between $B$ and $\iwarp{u}{t}$ in $\Imp_t$ must be larger than~$r$. Indeed, since $z'\in \tup z'$, there exists $u'\in \tup p(z')$ such that $\iwarp{u}{t}\neq \iwarp{u'}{t}$. Now if it was the case that $\dist_t(B,\iwarp{u}{t})\leq r$, then $u$ and $u'$ would be for sure colored with different colors in the $r$-distance coloring, contrary to $u,u'\in \tup p(z')$.

Now, consider any $\tup b \in \tups_{\tup p}$ and $\tup a\in A^{\tup x}$ such that $\tup a(m)\in A'$.
Let $U$ be the image of the tuple $\tup a\tup b$ (that is, the set of all vertices participating in this tuple).
Let $S$ be the unique inclusion-wise minimal subset of $U$ satisfying the following:
\begin{itemize}
    \item $\tup a(m) \in S$; and
    \item for every $v\in S$ and $u\in U$, if $\dist_t(\iwarp vt,\iwarp ut)\leq 2^k$, then $u \in S$.
\end{itemize}
Consider any $z'\in \tup z'$. Since $\tup b(z')\in \tup p(z')$, we have $\dist_t(B, \iwarp {\tup b(z')}t)>r$. Recalling that $\iwarp{\tup a(m)}t=B$ and $r=2^k(|\tup x|+|\tup y|-1)$, it follows from the definition of $S$ that $\tup b(z') \not \in S$.
Moreover, for every $v \in S$ and $u \in U \setminus S$, we have $\dist_t(\iwarp vt,\iwarp ut)>2^k$. Let $\tup q$, $\tup q'$ be the partition of the variable set $\tup x\tup y$ such that $\tup a\tup b(q)\in S$ for all $q\in \tup q$ and $\tup a\tup b(q')\in U\setminus S$ for all $q'\in \tup q'$. By the reasoning above, $m\in \tup q$ and $\tup z'\subseteq \tup q'$. 
By Lemma \ref{lem:ltp_compositionality}, we have that $\ltp^k_{t}(\tup a\tup b)$ depends only on $\ltp^k_{t}(\tup q)$ and $\ltp^k_{t}(\tup q')$.
As the local type of a subtuple depends only on the local type of the original tuple, and $\tup q\subseteq \tup x\tup z$ and $\tup q'\subseteq \tup x\tup y\setminus \{m\}$,
this concludes our proof.
\end{proof}

\subsection{Proof of the technical lemma}

We are ready to give a proof of  Lemma~\ref{lem:bounded_number_of_ltp}.

\begin{proof}[Proof of~Lemma~\ref{lem:bounded_number_of_ltp}]
 We proceed by induction on $|\tup y|$. The base case $|\tup y|=0$ is trivial, as then $|\tups_{\tup p}|=1$, so we assume $|\tup y|\geq 1$. By Lemma~\ref{lem:ltp_consistency}, we may assume that $t=\max_{y\in \tup y} \meet(\tup p(y))$. Note that for every $y\in \tup y$ there is a unique part $\tup v(y)\in \Pp_t$ such that $\tup p(y)\subseteq \tup v(y)$; thus $\tup v\in \Pp_t^{\tup y}$.
 
 We first consider the following special case. For a tuple $\tup u\in \Pp_t^{\tup x}$ such that $\tup u(x)\subseteq A$ for all $x\in \tup x$, define an equivalence relation $\hsim_{\tup u}$ on $\tups_{\tup p}$ as follows:
    \[
    \tup b \hsim_{\tup u} \tup b'\qquad\textrm{if and only if}\qquad  \ltp_t^k(\tup a\tup b) = \ltp_t^k(\tup a\tup b')\textrm{ for all }\tup a\in A^{\tup x}\textrm{ such that }\iwarp{\tup a}{t}=\tup u.
    \]
 Note that for each $\tup u$ as above, $\hsim_{\tup u}$ is a coarsening of $\sim$, as we require the equality of types of $\tup a\tup b$ and $\tup a\tup b'$ only for a certain subset of tuples $\tup a\in A^{\tup x}$. Our first goal is to prove the following claim.
 
 \begin{claim}\label{cl:special_case}
  For every fixed $\tup u\in \Pp_t^{\tup x}$ such that $\tup u(x)\subseteq A$ for all $x\in \tup x$, the equivalence relation $\hsim_{\tup u}$ has $\Oh_{d,k,\tup x,\tup y}(1)$ equivalence classes.
 \end{claim}
 \begin{proof}
Let us enumerate the variables of $\tup x$ as
$\seq x {|\tup x|}$ so that
$$\meet(\tup u(x_1))\leq \meet(\tup u(x_2))\leq \dots \leq \meet(\tup u(x_{|\tup x|})).$$
For brevity, for $i\in [|\tup x|]$ denote $\tup x_i\coloneqq \{x_1,\ldots,x_i\}$, $t_i\coloneqq \meet(\tup u(x_i))$, and $s\coloneqq t_{|\tup x|}$. Note that $t_1\leq t_2\leq \ldots\leq t_{|\tup x|}=s\leq t$. We may assume that $s>1$, for otherwise $\tup u$ is a tuple of singletons and there is only one tuple $\tup a$ with $\iwarp{\tup a}{t}=\tup u$, implying that $\hsim_{\tup u}$ has at most $|\Types^k_{\tup u\tup v,t}|$ equivalence classes, which is of $\Oh_{d,k,\tup x,\tup y}(1)$ by Lemma~\ref{lem:ltp_number_of_types}.

Observe that since $t=\meet(\tup p(y))$ for some $y\in \tup y$, the part of $\Pp_t$ that was obtained from contracting two parts of $\Pp_{t-1}$ must be contained in $V(G)\setminus A$. Similarly, since $s=\meet(\tup u(x_{|\tup x|}))$, the part of $\Pp_s$ that was obtained from contracting two parts of $\Pp_{s-1}$ must be contained in $A$. In particular it must be the case that  $s\neq t$. We conclude that $s<t$.

For every $i\in [|\tup x|]$, let $\tup y_i\subseteq \tup y$ be the set of those variables $y\in \tup y$ for which we have $\meet(\tup p(y))\leq t_i$. Since $t_{|\tup x|}=s<t$ and $t=\max_{y\in \tup y} \meet(\tup p(y))$, we have $\tup y_1 \subseteq \ldots \subseteq \tup y_{|\tup x|}\subsetneq \tup y$.

For a moment fix a tuple $\tup b\in \tups_{\tup p}$ and
consider any $\tup a\in A^{\tup x}$ such that $\iwarp{\tup a}{t}=\tup u$. By Lemma~\ref{lem:ltp_consistency} we have that $\ltp^k_{t}(\tup a \tup b)$ depends only on  
$\ltp^k_{t_{|\tup x|}}(\tup a \tup b)$. Next, since $t_{|\tup x|}=\meet(\tup u(x_{|\tup x|}))$,
by Lemma \ref{lemma:technical_lemma} we infer that $\ltp^k_{t_{|\tup x|}}(\tup a \tup b)$ depends only on the types $\ltp^k_{t_{|\tup x|}}(\tup a_{\tup x_{|\tup x|-1}}\tup b)$ and $\ltp^k_{t_{|\tup x|}}(\tup a\tup b_{\tup y_{|\tup x|}})$. Further, by Lemma~\ref{lem:ltp_consistency}, $\ltp^k_{t_{|\tup x|}}(\tup a_{\tup x_{|\tup x|-1}}\tup b)$ depends only on $\ltp^k_{t_{|\tup x|-1}}(\tup a_{\tup x_{|\tup x|-1}}\tup b)$.
By iterating this argument, we obtain that $\ltp^k_{t}(\tup a\tup b)$ depends only on the following types:
$$\ltp^k_{t_j}(\tup a_{\tup x_j}\tup b_{\tup y_j})\textrm{ for every }j \in [{|\tup x|}],\quad\textrm{ and } \ltp^k_{1}(\tup b).$$
In other words, for every $\tup b\in \tups_{\tup p}$ there is a function 
$$f_{\tup b}\colon \Types^k_{\tup u_{\tup x_1}\tup v_{\tup y_1},t_1}\times \ldots \times \Types^k_{\tup u_{\tup x_{|\tup x|}}\tup v_{\tup y_{|\tup x|}},t_{|\tup x|}}\to \Types^k_{\tup u\tup v,t}$$ such that for every $\tup a\in A^{\tup x}$ with $\tup u=\iwarp{\tup a}{t}$, we have 
\begin{equation}\label{eq:nutria}
f_{\tup b}(\alpha_1,\ldots,\alpha_{|\tup x|})=\ltp^k_t(\tup a\tup b),\qquad \textrm{where }\alpha_j=\ltp^k_{t_j}(\tup a_{\tup x_j}\tup b_{\tup y_j}).\end{equation} 
Observe that by Lemma~\ref{lem:ltp_number_of_types}, the number of different functions $f_{\tup b}$ as above is bounded by some constant $q\in \Oh_{d,k,\tup x,\tup y}(1)$.

For every $j\in [|\tup x|]$, let $\sim_j$ be the equivalence relation on $\tups_{\tup p}$ defined as follows:
\[
    \tup b \sim_j \tup b' 
    \qquad\textrm{if and only if}\qquad \ltp_{t_j}^k(\tup a_{\tup x_j}\tup b_{\tup y_j}) = \ltp_{t_j}^k(\tup a_{\tup x_j}\tup b'_{\tup y_j})\textrm{ for all }\tup a \in A^{\tup x}.
    \]
Since $|\tup y_j|<|\tup y|$ for all $j\in [|\tup x|]$, we may apply the induction assumption to infer that each equivalence relation $\sim_j$ has at most $c$ equivalence classes, where $c\in \Oh_{d,k,\tup x,\tup y}(1)$ is the bound obtained in the previous induction step.
Finally, define equivalence relation $\hsim^\star$ on $\tups_{\tup p}$ as follows:
\[
\tup b\hsim^\star \tup b'\qquad \textrm{if and only if}\qquad f_{\tup b}=f_{\tup b'}\ \textrm{ and }\ \tup b\sim_j\tup b' \textrm{ for all }j\in [|\tup x|].
\]
Clearly, $\hsim^\star$ has at most $q\cdot c^{|\tup x|}$ equivalence classes. It remains to observe that by~\eqref{eq:nutria}, $\hsim^\star$ refines~$\hsim$, so also $\hsim$ has at most $q\cdot c^{|\tup x|}$ equivalence classes.
\cqed\end{proof}

With Claim~\ref{cl:special_case} established, we proceed to the induction step in the general case. The proof is by a second induction, this time on $|\tup x|$. For $|\tup x| = 0$ the statement clearly follows from Lemma~\ref{lem:ltp_number_of_types}, so assume~$|\tup x| \ge 1$.

Let $\Ff$ be the set of all tuples $\tup u\in \Pp_t^{\tup x}$ such that $\tup u(x) \subseteq A$ and $\dist_t(\tup u(x),\tup v)\leq 2^k|\tup x|$ for all $x\in \tup x$. Clearly, we have $|\Ff|\leq \Oh_{d,k,\tup x,\tup y}(1)$. For each $\tup u\in \Ff$ we may consider the equivalence relation $\hsim_{\tup u}$, and by Claim~\ref{cl:special_case}, $\hsim_{\tup u}$ has $\Oh_{d,k,\tup x,\tup y}(1)$ equivalence classes.

Further, for every $\tup x'\subsetneq \tup x$, let $\sim_{\tup x'}$ be the equivalence relation defined in the same way as~$\sim$, but with respect to $\tup x'$-tuples. By the induction assumption, the number of equivalence classes of $\sim_{\tup x'}$ is $\Oh_{d,k,\tup x,\tup y}(1)$.

We now define the equivalence relation $\sim^\star$ on $\tups_{\tup p}$ as follows:
\[\tup b\sim^\star \tup b'\qquad\textrm{if and only if}\qquad \tup b\hsim_{\tup u}\tup b'\textrm{ for all }\tup u\in \Ff \textrm{ and }\tup b\sim_{\tup x'} \tup b'\textrm{ for all }\tup x'\subsetneq \tup x.\]
Clearly, $\sim^\star$ again has $\Oh_{d,k,\tup x,\tup y}(1)$ equivalence classes. So it remains to show that $\sim^\star$ refines $\sim$. That is, we need to show that for all $\tup b,\tup b'\in \tups_{\tup p}$, if $\tup b\sim^\star \tup b'$ then $\tup b\sim \tup b'$.

Let us then fix any $\tup b,\tup b'\in \tups_{\tup p}$ such that $\tup b\sim^\star \tup b'$. Consider any $\tup a\in A^{\tup x}$; our goal is to show that 
\begin{equation}\label{eq:morswin}\ltp^k_t(\tup a\tup b)=\ltp^k_t(\tup a\tup b').
\end{equation}
Let $\tup u \coloneqq \iwarp{\tup a}{t}$. If we have $\tup u\in \Ff$, then~\eqref{eq:morswin} follows from $\tup b\hsim_{\tup u}\tup b'$. Therefore, we may henceforth assume that $\tup u\notin \Ff$, or in other words, there exists $m\in \tup x$ such that $\dist(\tup u(m),\tup v)>2^k|\tup x|$. From this it follows that we may partition $\tup x$ into two subsets $\tup x'$ and $\tup x''$ so that $m\in \tup x''$ and for all $x\in \tup x''$ and $z\in \tup x'\tup y$, we have $\dist(\tup u(x),\tup u\tup v(z))>2^k$.
By Lemma~\ref{lem:ltp_compositionality} we infer that $\ltp^k_t(\tup a\tup b)$ depends only on $\ltp^k_t(\tup a_{\tup x''})$ and $\ltp^k_t(\tup a_{\tup x'}\tup b)$, and the same for $\tup b'$. But $\ltp^k_t(\tup a_{\tup x'}\tup b)=\ltp^k_t(\tup a_{\tup x'}\tup b')$, because $\tup x'\subsetneq \tup x$ and $\tup b\sim_{\tup x'} \tup b'$. So~\eqref{eq:morswin} holds and we are done.
\end{proof}

\bibliographystyle{abbrv}
\bibliography{ref}

\newpage
\appendix
\section{Effective computation of functions between local types}
\label{app:effective}
\newcommand{\trim}{\textsf{trim}}
\newcommand{\join}{\textsf{join}}
\newcommand{\promote}{\textsf{promote}}

Let $G$ be a graph, $\Pp = (\Pp_1,\ldots, \Pp_n)$ a contraction sequence of $G$ and $\tup x$ and $\tup y$ disjoint finite sets of variables.

For any $k\in \N$ and $\tup u \in \Pp_s^{\tup x}$, $\tup v \in \Pp_s^{\tup y}$ with $\dist_s(\tup u, \tup v) > 2^k$ we will define:
\begin{itemize}
\item Operation $\trim_{\tup u, s}^k: \Types_{\tup u, s}^k \to \Types_{\tup{u}, s}^{k-1}$ such that for any $\tup a \in V^{\tup x}$ with $\tup u = \iwarp{\tup a}{s}$ it holds that $\ltp_s^{k-1}(\tup{a}) = \trim_{\tup u, s}^k(\ltp_s^k(\tup a))$.
\item Operation $\join_{\tup{u},\tup{v}, s}^k: \Types_{\tup u, s}^k \times \Types_{\tup v, s}^k \to \Types_{\tup{uv}, s}^k$  such that for any $\tup a \in V^{\tup x}$, $\tup b \in V^{\tup y}$ with $\tup u = \iwarp{
\tup a}{s}$, $\tup v = \iwarp{\tup b}{s}$ it holds that $\ltp_s^k(\tup{ab}) = \join_{\tup{u},\tup{v}, s}^k(\ltp_s^k(\tup a), \ltp_s^k(\tup b))$.
\item Operation $\promote_{\tup u, s}^k: \Types_{\tup u, s}^k \to \Types_{\tup{u}, s+1}^{k}$ such  that for any $\tup a \in V^{\tup x}$ with $\tup u = \iwarp{\tup a}{s}$ it holds that $ \ltp_{s+1}^{k}(\tup{a}) = \promote(\ltp_s^k(\tup a))$.
\end{itemize}

After defining these operations, we prove their correctness and argue (in Section~\ref{app:tables}) that the functions they define can be computed in constant time for any given $k$, $d$ and $\tup x$, where $d$ is the width of the contraction sequence $\Pp = (\Pp_1,\ldots, \Pp_n)$. We remark that we are not merely interested in showing that these operations can be evaluated in constant time on any given input, but in showing that the whole input-output `table' of each operation can be constructed in constant time.

\subsubsection*{Operation $\trim$}
If $S$ is an atomic type over $\tup xz$ and $\tup uv \in \Pp_s^{\tup xz}$, then $\trim_{\tup uv, s}^0((S,\tup uv)):=(S', \tup u)$, where $S'$ is obtained from $S$ by removing all formulas containing variable $z$.
Let $s \in [n]$ be a time with $s\ge1$ and $\tup u \in \Pp_s^{\tup x}$. For $k\ge 1$ we define
$$ \trim_{\tup u, s}^k(\alpha) :=\setof{\trim_{\tup u, s}^{k-1}(\beta)}{\beta \in \alpha, \beta \in \Types_{\tup uv, s}^{k-1} \text{ and } \dist_s(\tup u,v)\le 2^{k-2}}.$$

\begin{lemma}\label{lem:trim}
Let $s \in [n]$ be a time, $\tup u \in \Pp_s^{\tup x}$ and $k\ge 1$. For any  $\tup a \in V^{\tup x}$ with $\tup u = \iwarp{\tup a}{s}$ it holds that $\ltp_s^{k-1}(\tup{a}) = \trim_{\tup u, s}^k(\ltp_s^k(\tup a))$.
\end{lemma}
\begin{proof}
By induction on $k$. For $k=1$ we have that any $\ltp^1_{\tup u, s}(\tup a)$ is a set of pairs $(S,\tup uv)$, where $S$ is an atomic type over $\tup xz$ and $S$ restricted to variables from $x$ is the atomic type of $\tup a$ in $G$. For each such pair the operation $\trim_{\tup uv, s}^0$ returns the pair  $(S', \tup u)$, where $S'$ is obtained from $S$ by removing all formulas involving variable $z$, which means that $S'$ is the atomic type of $\tup a$ in $G$, as desired.

For $k>1$, we will show that $\ltp_s^{k-1}(\tup{a}) \subset \trim_{\tup u, s}^k(\ltp_s^k(\tup a))$ and then $\ltp_s^{k-1}(\tup{a}) \supseteq \trim_{\tup u, s}^k(\ltp_s^k(\tup a))$.
Let $\alpha \in \ltp_{\tup u, s}^{k-1}(\tup a)$. Then $\alpha = \ltp^{k-2}_{\tup uv, s}(\tup ab)$ for some $b \in v$, where $v$ is a part of $\Pp_s$ with $\dist_s(\tup u,v) \le 2^{k-2}$. Since $\dist_s(\tup u,v) \le 2^{k-2}$, we have that  $\ltp^{k-1}_{\tup uv, s}(\tup ab) \in \ltp_{\tup u, s}^{k}(\tup a)$, and by applying the induction hypothesis we know that $\alpha = \ltp^{k-2}_{\tup uv, s}(\tup ab) =\trim_{\tup uv, s}^{k-1}(\tup ab)$, and therefore $\alpha \in \trim_{\tup u, s}^{k}(\ltp_{\tup u, s}^{k}(\tup a))$

To show that $\ltp_s^{k-1}(\tup{a}) \supseteq \trim_{\tup u, s}^k(\ltp_s^k(\tup a))$, let $\alpha \in  \trim_{\tup u, s}^k(\ltp_s^k(\tup a))$. Then $\alpha = \trim_{\tup uv, s}^{k-1}(\tup ab)$ for some $b \in v$, where $\dist_s(\tup u, v) \le 2^{k-2}$. By induction hypothesis $\trim_{\tup uv, s}^{k-1}(\tup ab) = \ltp_{\tup uv,s}^{k-2}(\tup a)$. Since $\dist_s(\tup u, v) \le 2^{k-2}$ and $b \in v$, by the definition of local types it follows that $ \ltp_{\tup uv,s}^{k-2}(\tup ab)$ is in $\ltp_{\tup u, s}^{k-1}(\tup a)$, which finishes the proof.
\end{proof}


\subsubsection*{Operation $\join$}
For two $0$-types $\alpha \in \Types_{\tup u, s}^0$ and $\beta \in \Types_{\tup v,s}^0$ such that $\dist_s(\tup u, \tup v) > 2^0$ we define 
\linebreak[4] $\join_{\tup{u},\tup{v}, s}^0((S_1, \tup u),(S_2, \tup v)):=(S, \tup{uv})$, 
where 
\begin{equation*} \label{eq1}
\begin{split}
S  := & S_1 \cup S_2 \cup \setof{x \not= y}{x \in \tup x, y \in \tup y}  \\ 
 & \cup   \setof{E(x,y)}{x \in \tup x, y \in \tup y \text{ and $u(x)$ is  adjacent to $v(u)$ in $G_s$}}   \\
 & \cup  \setof{\lnot E(x,y)}{x \in \tup x, y \in \tup y \text{ and $u(x)$ is not adjacent to $v(u)$ in $G_s$}}.
\end{split}
\end{equation*}

Let $z$ be a variable not in $\tup x$ or $\tup y$. For $k>0$ and two $0$-types $\alpha \in \Types_{\tup u, s}^0$ and $\beta \in \Types_{\tup v,s}^0$ such that $\dist_s(\tup u, \tup v) > 2^k$ we define 
\begin{equation*}
\begin{split}
\join_{\tup{u},\tup{v}, s}^k(\alpha, \beta):= &\setof{\join_{\tup{u}w,\tup{v}, s}^{k-1}(\gamma, \trim^k_{\tup v, s}(\beta))}{\gamma \in \Types_{\tup uw,s}^{k-1}, \tup uw\in \Pp_s^{\tup x\cup z} \text{ and } \dist_s(\tup u, w)\le 2^{k-1}}  \\
& \cup  \setof{\join_{\tup{u},\tup{v}w, s}^{k-1}(\trim^k_{\tup u, s}(\alpha), \gamma)}{\gamma \in \Types_{\tup vw,s}^{k-1}, \tup vw\in \Pp_s^{\tup y\cup z} \text{ and } \dist_s(\tup u, w)\le 2^{k-1}}
\end{split}
\end{equation*}

\begin{lemma}\label{lem:compose}
Let $\tup x$ and $\tup y$ be disjoint finite set of variables and let $s \in [n]$ be a time.
For any $k\in \N$ and $\tup u \in \Pp_s^{\tup x}$, $\tup v \in \Pp_s^{\tup y}$ with $\dist_s(\tup u, \tup v) > 2^k$ and tuples
$\tup a \in V^{\tup x}$, $\tup b \in V^{\tup y}$ with $\tup u = \iwarp{
\tup a}{s}$, $\tup v = \iwarp{\tup b}{s}$ it holds that $\ltp_s^k(\tup{ab}) = \join_{\tup{u},\tup{v}, s}^k(\ltp_s^k(\tup a), \ltp_s^k(\tup b))$.
\end{lemma}

\begin{proof}
By induction on $k$. For $k=0$, the local $0$-type of $\tup{ab}$  is $(S,\tup{uv})$, where $S$ is the atomic type of $\tup{ab}$ in $G$. Note that we have $\tup{ab} \in V(G)^{\tup{xy}}$. The adjacency between $\tup{ab}(x)$ and $\tup{ab}(x')$ for any $x,x' \in \tup x$ is determined by the atomic type of $\tup a$ in $G$, and analogously, the adjacency between $\tup{ab}(y)$ and $\tup{ab}(y')$ for any $y,y' \in \tup y$ is determined by the atomic type of $\tup b$ in $G$.
Since $\dist_s(\tup u, \tup v) > 2^0 = 1$, it holds that for any $x \in \tup x$ and $y \in \tup y$ the pair $\tup u(x)$, $\tup v(y)$ is pure, and so the adjacency between $\tup{ab}(x)$ and $\tup{ab}(y)$ for any $x \in \tup x$ and $y \in \tup y$ is the same as the adjacency between $\tup u(x)$, $\tup v(y)$ in $G_s$. This is exactly what the definition of $\join_{\tup{u},\tup{v}, s}^0$ describes.

Let $k > 0$. We will prove that $\ltp_s^k(\tup{ab}) \subset \join_{\tup{u},\tup{v}, s}^k(\ltp_s^k(\tup a), \ltp_s^k(\tup b))$ and \linebreak[4] $\ltp_s^k(\tup{ab}) \supseteq \join_{\tup{u},\tup{v}, s}^k(\ltp_s^k(\tup a), \ltp_s^k(\tup b))$. 
Let $\alpha \in \ltp_s^k(\tup{ab})$. 
Then $\alpha = \ltp_s^{k-1}(\tup{ab}c)$ for some $c \in w$, where $\dist(\tup{uv},w) < 2^{k-1}$ and $\tup uv \in \Pp_s^{\tup x \cup z}$. 
Since  $\dist_s(\tup u, \tup v) > 2^k$, exactly one of $\dist(\tup{u},w) < 2^{k-1}$ and $\dist(\tup{v},w) < 2^{k-1}$ has to hold. 
Assume it is the former, the latter case will be analogous. Since $\dist(\tup{u},w) < 2^{k-1}$, it holds that $\ltp_s^{k-1}(\tup{a}c) \in \ltp_s^{k}(\tup{a})$. 
By Lemma~\ref{lem:trim} we have that $\ltp_s^{k-1}(\tup b) = \trim_{\tup v,s}^k(\tup b)$, and by induction hypothesis we have that  $$\ltp_{s}^{k-1}(\tup ac\tup b) = \join_{\tup{u}w,\tup{v}, s}^{k-1}(\ltp_s^{k-1}(\tup ac),\trim_{\tup v,s}^k(\tup b)).$$
By the definition of $\join_{\tup{u},\tup{v},s}^{k}$, this belongs to  $\join_{\tup{u},\tup{v}, s}^k(\ltp_s^k(\tup a), \ltp_s^k(\tup b))$, and since $\ltp_{s}^{k-1}(\tup ac\tup b) = \ltp_{s}^{k-1}(\tup a\tup bc)$, the claim follows.

To prove that $\ltp_s^k(\tup{ab}) \supseteq \join_{\tup{u},\tup{v}, s}^k(\ltp_s^k(\tup a), \ltp_s^k(\tup b))$, let $\alpha \in \join_{\tup{u},\tup{v}, s}^k(\ltp_s^k(\tup a), \ltp_s^k(\tup b))$. 
Then $\alpha = \join_{\tup{u}w,\tup{v}, s}^{k-1}(\ltp_s^{k-1}(\tup ac),\trim_{\tup v,s}^k(\tup b))$ for some $c \in w$, where $\dist_s(\tup u,w)\le 2^{k-1}$ or \linebreak[4] $\alpha = \join_{\tup{u},\tup{v}w, s}^{k-1}((\trim_{\tup v,s}^k(\tup a), \ltp_s^{k-1}(\tup bc))$ for some $c \in w$, where $\dist_s(\tup v,w)\le 2^{k-1}$. We focus on the former case; the latter is analogous. By Lemma~\ref{lem:trim} it holds that $\trim_{\tup v,s}^k(\tup b) = \ltp_{s}^{k-1}(\tup b)$ and from the induction hypothesis it follows that  $\alpha = \join_{\tup{u}w,\tup{v}, s}^{k-1}(\ltp_s^{k-1}(\tup ac),\ltp_{s}^{k-1}(\tup b)) = \ltp_{s}^{k-1}(\tup ac\tup b)$, which is the same as $\ltp_{s}^{k-1}(\tup a\tup bc)$ and which belongs to $\ltp_s^k(\tup{ab})$, as desired.
\end{proof}
 
%
%

\subsubsection*{Operation $\promote$}
Let $s \in [n-1]$ be a time and $\tup u \in \Pp_s^{\tup x}$.
For any $(S,\tup u) \in \Types_{\tup u, s}^0$, then we define
$$ \promote_{\tup u, s}^0 := (S, \warp{\tup u}{s}{s+1}) $$

For $k>0$ and $\alpha \in \Types_{\tup u,s}^k$ we define
\begin{equation*} \label{eq1}
\begin{split}
\promote_{\tup u, s}^k(\alpha)  := &  \setof{\promote_{\tup uw, s}^{k-1}(\beta)}{\beta \in \alpha} \cup  \\
 & \setof{\promote_{\tup uv, s}^{k-1}(\join^{k-1}(\alpha,\gamma))}{\gamma \in \Types_{v,s}^k,  v \in V(\Rel_s^{2^{k-1}}),  \\
 &  \dist_s(\tup u, v)>2^{k-1}, \dist_{s+1}(\tup u,v)\le2^{k-1}}.
\end{split}
\end{equation*}

\begin{lemma}
Let $s \in [n]$ be a time, $\tup u \in \Pp_s^{\tup x}$ and $k\ge 1$. For any  $\tup a \in V^{\tup x}$ with $\tup u = \iwarp{\tup a}{s}$ it holds  that $ \ltp_{s+1}^{k}(\tup{a}) = \promote(\ltp_s^k(\tup a))$.
\end{lemma}

\begin{proof}
By induction on $k$. The case $k=0$ follows immediately from the definition of local types and operation $\promote_{\tup uv, s}^0$.

Let $k>0$. We will show that $ \ltp_{s+1}^{k}(\tup{a}) \subset \promote_{\tup uv, s}^k(\ltp_s^k(\tup a))$ and 
$ \ltp_{s+1}^{k}(\tup{a}) \supseteq \promote_{\tup uv, s}^k(\ltp_s^k(\tup a))$. Let $\alpha \in \ltp_{s+1}^{k}(\tup{a})$. Then $\alpha = \ltp_{s+1}^{k-1}(\tup ab)$ for some $b \in w$ such that $\dist_{s+1}(\tup u,w) \le 2^{k-1}$. 
We distinguish two possibilities:
\begin{enumerate}
\item $\dist_s(\tup u, w) \le 2^{k-1}$. In this case it holds $\ltp_s^{k-1}(\tup ab) \in \ltp_s^k(\tup a)$. By induction hypothesis we have $\ltp_{s+1}^{k-1}(\tup ab) = \promote_{\tup uw, s}^{k-1}(\ltp_s^{k-1}(\tup ab))$, and so in this case $\alpha$ is in  $\promote(\ltp_s^k(\tup a))$.
\item $\dist_s(\tup u, w) > 2^{k-1}$. In this case we have $\ltp_{s}^{k-1}(\tup ab) = \join_{\tup u, w,s}^{k-1}(\ltp_s^{k-1}(\tup a),\ltp_s^{k-1}(b))$ by Lemma~\ref{lem:compose}, and by induction hypothesis we get that $$\ltp_{s+1}^{k-1}(\tup ab) = \promote_{\tup uv, s}^k(\join_{\tup u, w,s}^{k-1}(\ltp_s^{k-1}(\tup a),\ltp_s^{k-1}(b))).$$ From the definition of $\promote_{\tup uv, s}^k$ it then follows that  $\ltp_{s+1}^{k-1}(\tup ab)$ is in $\promote_{\tup u, s}^k(\alpha)$
\end{enumerate}

To show that $ \ltp_{s+1}^{k}(\tup{a}) \supseteq \promote_{\tup uv, s}^k(\ltp_s^k(\tup a))$, let $\alpha \in \promote_{\tup uv, s}^k(\ltp_s^k(\tup a))$. 
Then $\alpha$ is either $\promote_{\tup uw, s}^{k-1}(\beta)$ for some $\beta \in \alpha$, or of the form 
$\promote_{\tup uv, s}^{k-1}(\join^{k-1}(\trim^{k-1}(\alpha),\trim^{k-1}(\gamma)))$, where $\gamma$ is a local $k$-type of part $v \in \Pp_s$ with $\dist_s(\tup u, v)>2^{k-1}$ and $\dist_{s+1}(\tup u,v)\le2^{k-1}$.
In the first case $\beta$ is of the form $\ltp_s^{k-1}(\tup ab)$ for some $b \in w$, where $w$ is a part with $\dist_s(\tup u,w) \le 2^{k-1}$, and by induction hypothesis we have that have that $\promote_{\tup uw, s}^{k-1}(\beta) = \ltp_{s+1}^{k-1}(\tup ab)$, which means that $\promote_{\tup uw, s}^{k-1}(\beta) \in \ltp^{k}_{s+1}$, as desired. 
In the second case $\gamma$ is of the form $\ltp_{s}^{k}(b)$ for some $b \in v$ and since we have $\dist_s(\tup u, v)>2^{k-1}$, we can apply the $\join$ operation to $\ltp_{s}^{k-1}(\tup a)$ and  $\ltp_{s}^{k-1}(b)$ to obtain $\ltp_s^{k-1}(\tup{ab})$, and so we have $\join^{k-1}(\alpha,\gamma) = \ltp_s^{k-1}(\tup{ab})$. Finally, by induction hypothesis we have that $ \promote_{\tup uv, s}^{k-1}(\ltp_s^{k-1}(\tup{ab})) = \ltp_{s+1}^{k-1}(\tup{ab})$, and so $\promote_{\tup uv, s}^{k-1}(\join^{k-1}(\alpha,\gamma))$ is in $\ltp_{s+1}^{k}(\tup a)$, as desired.
\end{proof}

\subsection{Computing input-output tables for $\trim$, $\join$ and $\promote$}
\label{app:tables}
We briefly sketch how to compute the input-output tables for $\trim$, $\join$ and $\promote$ in constant time for relevant parameters. More precisely, we argue that:
\begin{itemize}
\item The input-output table for $\trim_{\tup u, s}^k: \Types_{\tup u, s}^k \to \Types_{\tup{u}, s}^{k-1}$  can be computed in time $\Oh_{d,k,\tup x}(1)$ from $\Vic_{s}^{2^k}(\tup u)$.
\item The input-output table for $\join_{\tup{u},\tup{v}, s}^k: \Types_{\tup u, s}^k \times \Types_{\tup v, s}^k \to \Types_{\tup{uv}, s}^k$ can be computed in time $\Oh_{d,k,\tup x,\tup y}(1)$ from vicinity  $\Vic_s^{2^k}(\tup u\tup v)$.
\item The input-output table for $\promote_{\tup u, s}^k: \Types_{\tup u, s}^k \to \Types_{\tup{u}, s+1}^{k}$  can be computed in time $\Oh_{d,k,\tup x}(1)$ from the sub-trigraph of $G_s$ induced by  $V(\Rel_{s}^{2^k}(\tup u)) \cup V(\Vic_s^{2^k}(\tup u))$. Note that in case when for every $y \in \tup x$ we have that $\tup v(y) \in \Rel_s^{2^k|\tup x|}$, the required sub-trigraph of $G_s$ is contained in $\Rel_s^{2^k(|\tup x|+1)}$, as needed in Lemma~\ref{lem:tp-warp}.
\end{itemize}

The arguments are similar in each case, and are outlined below. In what follows we suppress the indices when they are irrelevant.
\begin{itemize}
\item For each operation the set of inputs and outputs has size bounded by $\Oh_{d,k,\tup x}(1)$ (for $\trim$ and $\promote$) or by  $\Oh_{d,k,\tup x, \tup y}(1)$ (for $\join$) and also these sets can be computed in time  $\Oh_{d,k,\tup x}(1)$ or $\Oh_{d,k,\tup x,\tup y}(1)$ for $\join$. This follows from Lemma~\ref{lem:ltp_number_of_types}.
\item Since each input set has bounded size, we can go through each possible input, apply the relevant operation and record the output to form the input-output table. We only need to show that the evaluation can be done in time $\Oh_{d,k,\tup x}(1)$ or $\Oh_{d,k,\tup x,\tup y}(1)$ and can be done using the information in $\Vic_s^{2^k}(\tup u)$, $\Vic_s^{2^k}(\tup u \tup v)$ or $\Rel_{s}^{2^k}(\tup u)$, based on the operation used.
\item We first note that the size of  $\Vic_s^{2^k}(\tup u)$ and  $\Rel_{s}^{2^k}(\tup u)$ is bounded by $\Oh_{d,k,\tup x}(1)$  and the size of $\Vic_s^{2^k}(\tup u \tup v)$ is bounded by $\Oh_{d,k,\tup x,\tup y}(1)$. Also, in each part $w$ contained in  $\Vic_s^{2^k}(\tup u)$,  $\Rel_{s}^{2^k}(\tup u)$ or $\Vic_s^{2^k}(\tup u \tup v)$  the number of local  $(k-1)$ in types $\Types_{w,s}^{k-1}$ is bounded by $\Oh_{d,k,\tup x}(1)$.
\item To bound the runtime of evaluating operations $\trim$, $\join$ and $\promote$, we note that each input $k$-type $\alpha$ (or each pair of input $k$-types $\alpha, \beta$ for $\join$), each operation recurses either on members of $\alpha$ (or $\beta$), which are $(k-1)$-types, or on $(k-1)$-types from parts in $\Vic_s^{2^k}(\tup u)$, $\Vic_s^{2^k}(\tup u \tup v)$ or $\Rel_{s}^{2^k}(\tup u)$, and in each case there is a bounded number of them as already argued, and so the branching is bounded in terms of $\Oh_{d,k,\tup x}(1)$ or $\Oh_{d,k,\tup x,\tup y}(1)$. In each branch, an operation of with index $(k-1)$ is used, and by an inductive argument one can use a table computed in time  $\Oh_{d,k,\tup x}(1)$ or $\Oh_{d,k,\tup x,\tup y}(1)$ to obtain the result. The one exception to this is that in the definition of operation $\join^k$ we use operation $\trim^k$, but in this case we can construct the table for $\trim^k$ before constructing the table for $\join^k$. 
\item Finally, to see that the operation $\trim$, $\join$ and $\promote$ can be computed from the information contained in $\Vic_s^{2^k}(\tup u)$,  $\Vic_s^{2^k}(\tup u \tup v)$ and  $\Rel_{s}^{2^k(\tup u)}$, we note that since every recursive definition of our operations uses distance $2^{k-1}$, the distances in the recursive calls will never exceed $\sum_{i\le k-1} 2^{i} \le 2^{k}$.  The only nontrivial check is then that the sub-trigraph of $G_s$ induced by  $V(\Rel_{s}^{2^k}(\tup u)) \cup V(\Vic_s^{2^k}(\tup u))$ contains all necessary information for computing operation $\promote_{\tup u,s}^{2^k}$. The important part is that it contains $\Vic_s^{2^{k-1}}(\tup uv)$ for any part $v$ in $\Rel_{s}^{2^{k-1}}$, which is necessary for the operation $\join^{k-1}_{\tup uv,s}$ used in the definition of $\promote_{\tup u,s}^{2^k}$.
This finishes the proof outline.
\end{itemize}

\end{document}